%% file: final2_bw.tex
\documentclass[final]{siamltex}
\usepackage{pstricks}
\usepackage{amsfonts}
\usepackage{graphicx}
\title{Weight function in a bimaterial strip containing an interfacial crack and an imperfect interface. Application to Bloch-Floquet analysis in a thin inhomogeneous structure with cracks.}
\author{A Vellender\thanks{Institute of Mathematics and Physics, Aberystwyth University ({\tt asv09@aber.ac.uk}).} \and G S Mishuris\thanks{Institute of Mathematics and Physics, Aberystwyth University ({\tt ggm@aber.ac.uk}).} \and A B Movchan\thanks{Department of Mathematical Sciences, University of Liverpool ({\tt abm@liv.ac.uk}).}}

\begin{document}
\maketitle

\begin{abstract}
We define a weight function and analyse a problem of anti-plane shear in a bi-material strip containing a semi-infinite crack and an imperfect interface. We then present an asymptotic algorithm which uses the weight function to evaluate the coefficients in asymptotics of solutions to problems of wave propagation in a thin bi-material strip containing a periodic array of cracks situated at the interface between two materials.
\end{abstract}

\pagestyle{myheadings}
\thispagestyle{plain}
\markboth{Vellender, Mishuris, Movchan}{Weight function in a bimaterial strip}

\section{Introduction}
In this paper we address the problem of determining a weight function in a domain representing a bi-material strip containing a semi-infinite interfacial crack. Where the crack is not present the interface is considered {\em imperfect}, modelling a thin layer of adhesive between the materials.

Weight functions are mainly used to evaluate stress intensity factors for asymptotic representations near non-regular boundaries such as crack tips. Classically, symmetric weight functions for interfacial cracks in two-dimensional elasticity were studied by Hutchinson et. al \cite{Hutchinson} and Bueckner \cite{Bueckner}. In these classical works, weight functions were defined as the stress intensity factors corresponding to the point force loads applied to the faces of the crack. More recently, Willis and Movchan \cite{Willis} defined general weight functions as non-trivial singular solutions of a boundary value problem with zero tractions on the faces of the crack and unbounded elastic energy. Recently weight functions have been used to perform perturbation analysis of the crack front in \cite{PiccWeight} and to evaluate Lazarus-Leblond constants in \cite{LazLeb}. These works contain perfect interfaces which lead to the well-known square root singularity phenomenon \cite{Rice,Willis}. In the imperfect interface problem considered in the present paper there is no square root singularity in stress components and so the weight function instead takes the role of aiding in the evaluation of important asymptotic constants which take the place of stress intensity factors.

The imperfect interface is a crucial feature of the problem discussed.
\textcolor{black}{Accurate asymptotic derivations with various interfaces in composite materials (of imperfect type among others) for anti-plane shear without the presence of cracks have been analysed in \cite{Benveniste,Hashin,Lipton}. 
Such interfaces have been used to model a thin layer consisting of small cracks in such a way that the cracks do not appear in the analysis in \cite{Baik,Bostrom,GolubZhangWang} using the phenomenological approach.
Cracks in the static regime with imperfect interfaces have been studied in \cite{Antipov2001,Mish2001a}, where it is proved that the imperfect interface leads to a different type of singularity near the crack tip than in the ideal interface case.
Analysis of the perfect interface with cracks under harmonic load can be found in \cite{AchenbachLi,Mish2007} and recently for a layered composite with cracks in \cite{GolubZhangWang}. The manuscript \cite{Mish2007} considers wave propagation in a thin bi-material strip and discusses the singular behaviour near the crack tip, while \cite{AchenbachLi} considers a bi-material plane with waves propagating perpendicular to the cracks.}

We consider in this paper Mode III deformation and describe the extent of the interface's imperfection by a positive parameter denoted $\kappa$. The problem we study here is a singular perturbation problem; taking very small values for $\kappa$ gives a qualitatively significantly different weight function from that derived for the perfect interface case in \cite{Mish2007} which corresponds to the formulation with $\kappa=0$. Moreover, large values of $\kappa$ can lead to interesting effects where the boundary layers surrounding different crack tips decay slowly so they can no longer be considered as having no influence on the Bloch-Floquet conditions. This effect is discussed in \cite{Orlando2003}; in the analysis presented in the present paper we assume that $\kappa$ is not large enough for these effects to come into play and later find a condition for this to be the case. Problems regarding cracks in domains including imperfect interfaces have been studied in \cite{Antipov2001} and \cite{Mish2006}, but no corresponding weight function has previously been constructed.

\textcolor{black}{Another critical characteristic of the problem is that the strip considered is very thin. In addition to the strip itself being very thin, imperfect interfaces are typically replaced with an extremely thin layer of a softer bonding material in finite element computations (justified for example in \cite{Benveniste,Hashin,Mish2006}). Moreover, singular behaviour exists at the crack tips. These points make FEM modelling for particularly thin strips extremely difficult or even impossible and motivate the need for the asymptotic approach. In this paper we compare the asymptotic model with finite element simulations only in cases when the strip is not too thin, but stress that the finite element methods are unsuitable for the limiting case whereas the asymptotics remain valid. The asymptotic method also obtains crucial constants which describe the solution's behaviour at the crack tips which are vital for determining whether fracture may occur. These important constants would not be attained by finite element methods.}

The plan of the work is as follows. We first formulate the weight function problem and use Fourier transform and Wiener-Hopf techniques \cite{Noble} to obtain the solution. Asymptotic analysis enables us to find analytic expressions for all important constants. We then present an application of the weight function to the analysis of Bloch-Floquet waves in a structure containing a periodic array of cracks and imperfect interfaces. This application involves the derivation of junction conditions. 
\textcolor{black}{Asymptotic theories for structures like rods and plates have received much attention throughout the history of elasticity theory. For multi-structures however, conditions in engineering practice are often formulated on the basis of intuitive physical assumptions \cite{Asplund}. For example, the zero order junction conditions for the problem addressed fit with physical intuition. It is important to give these conditions a rigorous mathematical footing; moreover, higher order junction conditions do not follow such intuition \cite{Kozlov}. 
} 

We conclude by presenting a comparison between the perfect interface case studied in \cite{Mish2007} and the imperfect interface case presented here.

\section{Weight Function}\label{weightsection}
\subsection{Formulation of the Problem}
\begin{figure}[t]
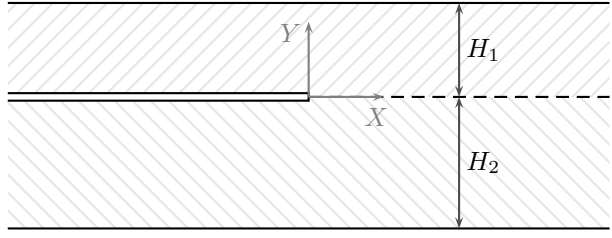

\begin{center}
\include{weightsetup}
\end{center}
\caption{Geometry for the weight function.}
\label{weightsetup}
\end{figure}
The geometry of the strip in which we construct the weight function is shown in Figure \ref{weightsetup}. We define our domain $\Pi_B$ to be the union of $\Pi_B^{(1)}$ and $\Pi_B^{(2)}$, where
\[
\Pi_B^{(j)}=\{(X,Y):X\in\mathbb{R}, \quad (-1)^{j+1}Y\in(0,H_j)\},\quad j=1,2.
\]
$\Pi_B^{(1)}$ corresponds to the material above the cut with shear modulus $\mu_1$, while $\Pi_B^{(2)}$ corresponds to the material below the cut with shear modulus $\mu_2$. The materials have respective thicknesses $H_1$ and $H_2$. A semi-infinite crack with its tip placed at the origin occupies $X<0$, while the rest of the interface is assumed to be imperfect (see (\ref{transcons1}) in the text below).

The functions $w_1$ and $w_2$ are defined in domains $\Pi_1$ and $\Pi_2$ respectively as solutions to the Laplace equation
\begin{equation}\label{laplace}
\nabla^2 w_j(X,Y)=0.
\end{equation}
We impose boundary conditions along the horizontal parts of the boundary of $\Pi_B$ and on the crack face itself. We denote the components of stress in the out-of-plane direction by
\begin{equation}\label{stressnotation}
 \sigma_{nz}^{(j)}(X,Y):=\mu_j\frac{\partial u^{(j)}}{\partial n},\quad j=1,2.
\end{equation}
We assume a zero stress component in the out-of-plane direction along the top and bottom of the strip, as well as along the face of the crack itself:
\begin{eqnarray}
\sigma_{YZ}^{(1)}(X,H_1)=0,\quad\sigma_{YZ}^{(2)}(X,-H_2)=0, \quad X\in\mathbb{R},\label{boundcons1}
\\\sigma_{YZ}^{(1)}(X,0^+)=0,\quad \sigma_{YZ}^{(2)}(X,0^-)=0, \quad X<0.\label{boundcons3}
\end{eqnarray}
Ahead of the cut we impose the imperfect transmission conditions
\begin{equation}\label{transcons1}
\left.w_1\right|_{Y=0_+}-\left.w_2\right|_{Y=0_-}=\kappa\sigma_{YZ}^{(1)}(X,0^+),\quad X>0,
\end{equation}
where $\kappa>0$ is a parameter describing the extent of imperfection of the interface. We further assume continuity of tractions across the interface between the materials
\begin{equation}\label{transcons2}
\sigma_{YZ}^{(1)}(X,0^+)=\sigma_{YZ}^{(2)}(X,0^-),\quad X>0.
\end{equation}

We seek solutions in the class of functions that decay exponentially as $X\to+\infty$ and are bounded as $X\to-\infty$:
\begin{equation}\label{watinf}
w_j=O(e^{-\gamma_+ X}),\quad X\to+\infty;\qquad w_j=C_j+O(e^{\gamma_- X}),\quad X\to-\infty,
\end{equation}
where $\gamma_\pm>0$ and $C_j$ are constants to be sought from the analysis. At the vertex of the crack, the solution $w_j$ is assumed to be weakly singular, with
\begin{equation}\label{wzero}
w_1, w_2=O(\ln|X|), \quad X\to0.
\end{equation}
Formally, conditions (\ref{laplace})-(\ref{watinf}) are similar to those in \cite{Mish2007} if we take $\kappa=0$. However, with $\kappa>0$ the problem is a singular perturbation problem and the behaviour described in (\ref{wzero}) is entirely different.

\subsection{An auxiliary problem}
We now introduce an auxiliary solution $\mathcal{Y}$ which satisfies the Laplace equation (\ref{laplace}) along with the boundary and transmission conditions (\ref{boundcons1})-(\ref{transcons2}), but the conditions at infinity and at the vertex of the crack are modified as follows:
\begin{eqnarray}
	\label{plusinf}\mathcal{Y}_j&=&O(e^{-\gamma_+X}), \quad X\to+\infty,
	\\\label{minusinf}\mathcal{Y}_j &=& C_jX+D_j+O(e^{\gamma_- X}), \quad X\to-\infty,
	\\\mathcal{Y}_j&=&\mathcal{Y}_j(0_+,0)+O(X\ln|X|),\quad X\to0.\label{2.20}
\end{eqnarray}

The functions $w$ and $\mathcal{Y}$ are related via the formula
\begin{equation}
w(X,Y)=\frac{\partial}{\partial X}\mathcal{Y}(X,Y).
\end{equation}
Bearing this relationship in mind, we often later refer to $\mathcal{Y}$ as a `weight function' as well as $w$.
It is also shown in \cite{Mish2001a} that as $R\to0$,
\begin{equation}\label{Mish2001aform}
 \mathcal{Y}_\pm=\frac{(-1)^ja_0^{(\mathcal{Y})}}{\pi\mu_j}\left\{\frac{\mu_1\kappa\pi}{1+\frac{\mu_1}{\mu_2}} + \left[1-\ln\left(\frac{R}{b_0^{(\mathcal{Y})}}\right)\right]R\cos\theta\pm(\pi\mp\theta)R\sin\theta\right\},
\end{equation}
where $\mathcal{Y}_+$ and $\mathcal{Y}_-$ represent $\mathcal{Y}_1(R,\theta)$ and $\mathcal{Y}_2(R,\theta)$ respectively and $(R,\theta)$ describes the usual polar co-ordinate system with $\theta\in[0,\pi]$ for $\mathcal{Y}_1$ and $\theta\in[-\pi,0]$ for $\mathcal{Y}_2$.

\subsubsection{Derivation of Wiener-Hopf equation}
We define the Fourier transforms of $\mathcal{Y}_j$ by
\begin{equation}
{\bar{\mathcal{Y}}_j}(\xi, Y)=\int\limits_{-\infty}^\infty{e^{i\xi X}\mathcal{Y}_j}(X,Y)dX.
\end{equation}
The functions ${\bar{\mathcal{Y}}_j}$ are analytic in the strip $S=\{\xi\in\mathbb{C}:-\gamma_+<\mathrm{Im}(\xi)<0\},$ and have a double pole only at the point $\xi=0$, so
\begin{equation}\label{pole}
{\bar{\mathcal{Y}}_j}(\xi,Y)\sim\frac{1}{\xi^2}C_j-i\frac{D_j}{\xi}+O(1), \quad\xi\to0.
\end{equation}
Note that the functions ${\bar{\mathcal{Y}}_j}(\xi,Y)$ can be analytically extended to the strip
\[
\tilde{S}=\{\xi\in\mathbb{C}:-\gamma_+<\mathrm{Im}(\xi)<\gamma_-\}.
\]
Let us now introduce $[\mathcal{Y}]$, the jump in $\mathcal{Y}$, defined by
\begin{equation}
[\mathcal{Y}]=\left.{\mathcal{Y}_1}\right|_{Y=0+}-\left.{\mathcal{Y}_2}\right|_{Y=0-}.
\end{equation}
We see from (\ref{pole}) that the Fourier transform of the jump $[\mathcal{Y}](X)$ generally speaking has a double pole at the point $\xi=0$.

We introduce the following notation:
\begin{equation}
\Phi^-(\xi)=\overline{[\mathcal{Y}]-\left.\mu_1\kappa\frac{\partial\mathcal{Y}_1}{\partial Y}\right|_{Y=0+}}=\int\limits_{-\infty}^0{\left([\mathcal{Y}](X)-\left.\mu_1\kappa\frac{\partial\mathcal{Y}_1}{\partial Y}\right|_{Y=0+}\right)e^{i\xi X}dX},
\end{equation}
where we have taken into account (\ref{transcons1}) or equivalently the fact that $[\mathcal{Y}]-\left.\mu_1\kappa\frac{\partial\mathcal{Y}_1}{\partial Y}\right|_{Y=0+}=0$ for $X>0$.
The function $\Phi^-(\xi)$ is analytic in the half plane $\mathrm{Im}(\xi)<0$ and has a double pole at $\xi=0$. Thus it can be analytically extended into the half-plane
$\mathbb{C}^-=\{\xi\in\mathbb{C}:\mathrm{ Im}(\xi)<\gamma_-\}.$
We further define the function
\begin{equation}\label{phiplusdef}
\Phi^+(\xi)=\mu_1\int\limits_0^\infty{\left.\frac{\partial \mathcal{Y}_1}{\partial Y}\right|_{Y=0+}e^{i\xi X}dX},
\end{equation}
and so according to (\ref{boundcons3}), $\Phi^+(\xi)$ is analytic in the half plane
$\mathbb{C}^+=\{\xi\in\mathbb{C}:\mathrm{ Im}(\xi)>-\gamma_+\}$.

We expect that
\begin{equation}\label{phiexpect}
 \Phi^\pm(\xi)=\frac{E_1^\pm}{\xi}+\frac{E_2^\pm\ln(\mp i\xi)}{\xi}+O\left(\frac{1}{\xi^2}\right),\quad\xi\to\infty,
\end{equation}
in the respective domain according to (\ref{2.20}); we later confirm this to be true. 

\textcolor{black}{The Fourier transforms of the functions $\mathcal{Y}_j$ are of the form
\begin{equation}\label{fourieryj}
\bar{\mathcal{Y}}_j(\xi,Y)=A_j(\xi)\cosh(\xi Y)+B_j(\xi)\sinh(\xi Y).
\end{equation}
Upon the application of boundary and transmission conditions expressions relating $A_j(\xi)$ and $B_j(\xi)$ are found:
\begin{equation}\label{ajbj}
B_j(\xi)=(-1)^j A_j(\xi)\tanh(\xi H_j),\quad j=1,2; \qquad \mu_1B_1(\xi)-\mu_2B_2(\xi)=0.
\end{equation}
Moreover, $\Phi^\pm(\xi)$ can be expressed in terms of $A_j(\xi), B_j(\xi)$.
\begin{equation}\label{phipm}
\Phi^-(\xi)=A_1(\xi)-A_2(\xi)-\mu_1\kappa\xi B_1(\xi),\quad \Phi^+(\xi)=\mu_1\xi B_1(\xi).
\end{equation}
By applying boundary and transmission conditions, we conclude that the functions $\Phi^+(\xi)$ and $\Phi^-(\xi)$ satisfy the functional equation of the Wiener-Hopf type
\begin{equation}\label{wh}
\Phi^-(\xi)=-\Xi(\xi)\Phi^+(\xi),
\end{equation}
}in the strip $-\gamma_+<\mathrm{Im}(\xi)<0$, where
\begin{equation}\label{bigxi}
\Xi(\xi)=\frac{1}{\xi}\left(\frac{1}{\mu_1}\coth(\xi H_1)+\frac{1}{\mu_2}\coth(\xi H_2)+\kappa\xi\right),
\end{equation}
and $-\gamma_+$ is equal to the size of the imaginary part of the first zero of $\Xi(\xi)$ lying below the real axis. We would like to stress that the form of the Wiener-Hopf kernel $\Xi(\xi)$ demonstrates that the weight function problem is a singular perturbation problem as $\kappa\to0$; the presence of the term involving $\kappa$ fundamentally alters the asymptotic behaviour of $\Xi(\xi)$ as $\xi\to\infty$.

\subsubsection{Factorization of the Wiener-Hopf kernel}
We note that the kernel function $\Xi(\xi)$ as defined in (\ref{bigxi}) can be written in the form
\begin{equation}\label{factorisedeqn}
\Xi(\xi)=\kappa\frac{(\lambda+i\xi)(\lambda-i\xi)}{\xi^2}\Xi_*(\xi),
\end{equation}
where
\begin{equation}
\Xi_*(\xi)=\frac{\xi(\mu_1\coth(\xi H_2)+\mu_2\coth(\xi H_1)+\mu_1\mu_2\kappa \xi)}{\mu_1\mu_2\kappa(\lambda^2+\xi^2)},
\end{equation}
and
\begin{equation}\label{lambdadef}
 \lambda=\sqrt{\frac{\mu_1H_1+\mu_2H_2}{\mu_1\mu_2H_1H_2\kappa}}.
\end{equation}

Now, $\Xi_*(\xi)$ is analytic in a strip containing the real axis, clearly positive, even and smooth for all $\xi\in\mathbb{R}$ and has been chosen in such a way so that $\Xi_*(\xi)$ tends towards 1 as $\xi\to\pm\infty$ and as $\xi\to0$. Furthermore, the function $\Xi_*(\xi)$ can be factorized in the form
\begin{equation}\label{xistarfactorised}
\Xi_*(\xi)=\Xi_*^+(\xi)\Xi_*^-(\xi),
\end{equation}
where
\begin{equation}\label{xistarpm}
\Xi_*^\pm(\xi)=\exp\left\{{\frac{\pm1}{2\pi i}\int\limits^{\infty\mp i\beta}_{-\infty\mp i\beta}{\frac{\ln\Xi_*(t)}{t-\xi}dt}}\right\},
\end{equation}
and $\beta>0$ is chosen to be sufficiently small so the contours of integration lie within the strip of analyticity of $\Xi_*(\xi)$. The functions $\Xi_*^\pm$ are analytic in their respective half-planes.
To conclude this subsection, we have factorised $\Xi(\xi)$ in the form given in (\ref{factorisedeqn}) and (\ref{xistarfactorised}),
where $\Xi_*^\pm$ are analytic in the half-planes denoted by their superscripts. Note that in the case $H_1=H_2$, other factorisation has been obtained in \cite{Antipov2001}.

\subsubsection{Asymptotic behaviour of $\Xi_*^+$}
We now seek asymptotic estimates of $\Xi_*^+(\xi)$. We first note that for $\xi$ within the strip of analyticity,
\begin{equation}\label{estofxi0}
\Xi(\xi)=\frac{\eta}{\xi^2}+O(1), \quad \Xi_*(\xi)=1+O(|\xi|^2), \quad\xi\to0,\quad \eta=\frac{1}{\mu_1 H_1}+\frac{1}{\mu_2 H_2}.
\end{equation}
Let us now consider more accurately the behaviour of $\Xi_*(\xi)$ for $\xi\in\mathbb{R}$ as $\xi\to\infty$. Noting that $\Xi_*(\xi)$ is an even function, it follows from (\ref{bigxi}) that

\begin{equation}\label{bigxiasym}
\Xi_*(\xi)=1+\frac{\mu_1+\mu_2}{\mu_1\mu_2\kappa|\xi|}-\frac{\lambda^2}{\xi^2}+O\left(\frac{1}{|\xi|^3}\right),\quad\xi\to\pm\infty.
\end{equation}
The same estimate is true for any $\xi$ lying in the strip of analyticity. We further find that
\begin{equation}\label{xistarplus0}
\Xi_*^+(\xi)
=1+\frac{\alpha\xi}{\pi i}+O(|\xi|^2),\quad\xi\to0,
\end{equation}
\begin{equation}\label{xistarplusinf}
\Xi_*^+(\xi)=1+\frac{1}{\pi i}\frac{(\mu_1+\mu_2)}{\mu_1\mu_2\kappa}\frac{\ln(-i\xi)}{\xi}+O\left(\frac{1}{|\xi|}\right),\quad\mathrm{Im}(\xi)\to+\infty;
\end{equation}
the derivation of these expressions is given in Appendix A. Here we have defined the asymptotic constant
\begin{equation}\label{alphadef}
 \alpha=\int\limits_{0}^\infty{\frac{\ln\Xi_*(t)}{t^2}dt}.
\end{equation}
The important expression (\ref{xistarplusinf}) describing logarithmic asymptotics at infinity is needed later for equation (\ref{phiplusverified}).

\subsubsection{Solution of the Wiener-Hopf equation}

The factorized equation (\ref{wh}) is of the form
\begin{equation}\label{whsoln}
-\kappa(\lambda-i\xi)\Phi^+(\xi)\Xi_*^+(\xi)=\frac{1}{\lambda+i\xi}\xi^2\Phi^-(\xi)\frac{1}{\Xi_*^-(\xi)}.
\end{equation}
Both sides of (\ref{whsoln}) represent analytic functions in the strip $-\gamma_+<\mathrm{Im}(\xi)<\gamma_-$. Moreover we now have asymptotic estimates for $\Xi_*^\pm(\xi)$ at the zero point in equation (\ref{xistarplus0}) and for $\xi\to\pm\infty$ in (\ref{xistarplusinf}). We deduce that since both sides of (\ref{whsoln}) exhibit the same behaviour at infinity in their respective domains according to (\ref{phiexpect}), both sides must be equal to a constant, which we denote $\mathcal{A}$. We can therefore obtain explicit expressions for $\Phi^\pm$, which are as follows:
\begin{equation}\label{phis}
	\Phi^+(\xi)=-\frac{\mathcal{A}}{\kappa(\lambda-i\xi)\Xi_*^+(\xi)},\qquad
	\Phi^-(\xi)=\frac{\mathcal{A}(\lambda+i\xi)\Xi_*^-(\xi)}{\xi^2},
\end{equation}
We deduce that
\begin{equation}\label{ysoln}
	\bar{\mathcal{Y}}_j(\xi,Y)=-\frac{\mathcal{A}\Phi^+(\xi)}{\mu_j\xi}\left\{\frac{\cosh(\xi(Y+(-1)^j H_j))}{\sinh(\xi(-1)^{j+1}H_j)}\right\},\quad j=1,2.
\end{equation}
This allows us to investigate the behaviour of $\bar{\mathcal{Y}}_j$ as $\xi\to\pm\infty$ and at the zero point. It also enables us to find the hitherto unknown real constants $C_j$ and $D_j$.

\subsubsection{Evaluation of constants $C_j$, $D_j$, $a_0^{(\mathcal{Y})}$, $\gamma_\pm$}
In this subsection we evaluate the constants $\gamma_+$ (defined in (\ref{plusinf})), $\gamma_-$, $C_j$, $D_j$ (defined in (\ref{minusinf})) and $a_0^{(\mathcal{Y})}$ (defined in (\ref{Mish2001aform})). We see from our expressions for $\bar{\mathcal{Y}}_j$ and $\Phi^+$ (equations (\ref{phis}) and (\ref{ysoln})), along with our asymptotic estimate for $\Xi_*^+(\xi)$ as $\xi\to0$ that
\begin{equation}
	\bar{\mathcal{Y}}_j(\xi)=\frac{(-1)^{j+1}\mathcal{A}}{\kappa\lambda\mu_j H_j}\left(\frac{1}{\xi^2}-\frac{i}{\xi}\left(-\frac{\alpha}{\pi}-\frac{1}{\lambda}\right)\right)+O(1),\quad \xi\to0,
\end{equation}
where $\alpha$ is the constant defined in (\ref{alphadef}). It follows from our definition of $C_j$ and $D_j$ in (\ref{pole}) that
\begin{equation}
C_j=\frac{(-1)^{j+1}\mathcal{A}}{\kappa\lambda\mu_j H_j},\quad\quad D_j=\frac{(-1)^j\mathcal{A}}{\kappa\lambda\mu_j H_j}\left(\frac{\alpha}{\pi}+\frac{1}{\lambda}\right).
\end{equation}
For normalisation we choose $\mathcal{A}=\kappa\lambda$, giving
\begin{equation}
C_j=\frac{(-1)^{j+1}}{\mu_j H_j},\quad\quad D_j=\frac{(-1)^j}{\mu_j H_j}\left(\frac{\alpha}{\pi}+\frac{1}{\lambda}\right).\label{cjdj}
\end{equation}
The chosen normalisation leaves (\ref{cjdj})$_1$ in the same form as in \cite{Mish2007}, but it is clearly seen that the expression for $D_j$ (which depends upon $\kappa$ is different). Mishuris (2001) \cite{Mish2001a} demonstrates that near the crack tip (i.e. as $R\to0$), $\mathcal{Y}_j(R,\theta)$ has behaviour described by (\ref{Mish2001aform}). From this we see that
\begin{equation}
 [\mathcal{Y}]\sim-\kappa a_0^{(\mathcal{Y})},\quad R\to0.
\end{equation}
The imperfect transmission conditions (\ref{transcons1}) therefore give that
\begin{equation}
 \mu_1\left.\frac{\partial\mathcal{Y}_1}{\partial Y}\right|_{Y=0+}\sim-a_0^{(\mathcal{Y})},\quad X\to0.
\end{equation}

We earlier made an assumption in (\ref{phiexpect}) regarding the behaviour of $\Phi^+$ at infinity and now verify that this was correct. It follows from the expression for $\Phi^+(\xi)$ given in (\ref{phis}) and the asymptotic estimate for $\Xi_*^+(\xi)$ given in (\ref{xistarplusinf}) that
\begin{equation}\label{phiplusverified}
\Phi^+(\xi)=\frac{\lambda}{i\xi}+\frac{(\mu_1+\mu_2)\lambda}{\mu_1\mu_2\pi\kappa\xi^2}\ln(-i\xi)+O\left(\frac{1}{|\xi|^2}\right),\quad\mathrm{Im}(\xi)\to+\infty,
\end{equation}
which justifies our previous claim. Theorem \ref{theoremforphi} (using $\mu_1\frac{\partial \mathcal{Y}}{\partial Y}$ in place of `$f$' in the statement of the theorem) then yields that
\begin{equation}
 \lim\limits_{X\to0+}\mu_1\frac{\partial \mathcal{Y}}{\partial Y}=-\lambda,
\end{equation}
where $\lambda$ has been defined in (\ref{lambdadef}) and so it follows that
\begin{equation}\label{a0value}
 a_0^{(\mathcal{Y})}=\lambda.
\end{equation}
The constant $\gamma_+$ is the distance of the first zero of $\Xi(\xi)$ below the real axis. Manipulation of (\ref{bigxi}) indicates that zeros of $\Xi(\xi)$ satisfy
\begin{equation}
 \frac{1}{\mu_1}\cot(\gamma_+ H_1)+\frac{1}{\mu_2}\cot(\gamma_+ H_2)-\kappa\gamma_+=0,
\end{equation}
For the first zero below the axis, for large $\kappa$, $\gamma_+$ should be small, and so it can be shown that
\begin{equation}\label{gammaplusest}
 \gamma_+(\kappa)=\lambda(\kappa)(1+O(\kappa^{-1})),\quad\kappa\to\infty,
\end{equation}
indicating that $\gamma_+(\kappa)=O(\kappa^{-1/2})$, $\kappa\to\infty$. We also see that
\begin{equation}
 \gamma_+(0)\in\left(\frac{\pi}{2H_1},\frac{\pi}{2H_2}\right).
\end{equation}

The constant $\gamma_-$ is given by
\begin{equation}
 \gamma_-=\pi\min\left\{\frac{1}{H_1},\frac{1}{H_2}\right\}.
\end{equation}
In conjunction with (\ref{cjdj}) we have now found all constants describing the asymptotic behaviour of the weight function $\mathcal{Y}$.

\section{Application to Analysis of Bloch-Floquet Waves}\label{application}
In this section, we present an application of the weight function derived in the previous section by addressing the problem of out-of plane shear Bloch-Floquet waves within a thin bi-material strip containing a periodic array of longitudinal cracks and imperfect interfaces. The problem addressed is an imperfect interface analogue to that studied in \cite{Mish2007}.
\subsection{Geometry}
The geometry of an elementary cell of the thin periodic structure considered is shown in Figure \ref{cellgeom}. The elementary cell is of length $a$ and contains two materials of thicknesses $\varepsilon H_1$ and $\varepsilon H_2$, where $\varepsilon$ is a small dimensionless parameter. These materials occupy respective domains $\Pi^{(j)},$ $j=1,2,$ and the elementary cell is further split into smaller domains $\Omega_\varepsilon^{(m)},$ $m=1,2,3,4,$ as shown in Figure \ref{cellgeom}. Along the join of the two materials and centered on the origin sits a crack of length $l$. Outside the crack, the interface is assumed to be imperfect, which models a thin layer of adhesive joining the materials together. The extent of this imperfection is represented by the parameter $\kappa$.
\begin{figure}[t]
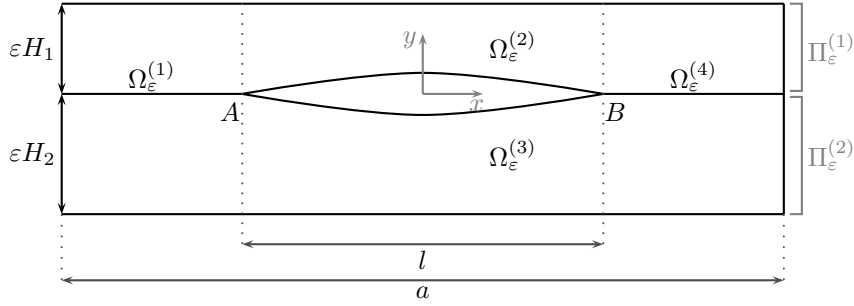

\begin{center}
\include{elementaryblock}
\end{center}
\caption{Geometry of the elementary cell.}
\label{cellgeom}
\end{figure}

The functions $u^{(j)}(x,y)$ are defined in $\Pi_\varepsilon^{(j)}$, $j=1,2$ as solutions of the Helmholtz equations
\begin{equation}\label{helmholtz}
 \nabla^2 u^{(j)}(x,y)+\frac{\omega^2}{c_j^2}u^{(j)}(x,y)=0,\qquad (x,y)\in\Pi_\varepsilon^{(j)},\quad j=1,2.
\end{equation}
 Here, $c_j=\sqrt{\mu_j/\rho_j}$ are the shear speeds in their respective domains $j=1,2$. The functions $u^{(j)}$ are regarded as out-of-plane displacements, $\mu_j$ denotes the shear modulus and $\rho_j$ the mass density of the material occupying $\Pi_\varepsilon^{(j)}$. The quantity $\omega$ represents the radian frequency of the time-harmonic vibrations with  amplitude $u$.

\subsection{Boundary conditions}
We impose boundary conditions along the horizontal parts of the boundary of $\Pi_\varepsilon$ and on the crack face itself. We use similar notation to that in the previous section to denote the components of stress (see (\ref{stressnotation})).

We assume a zero stress component in the out-of-plane direction along the top and bottom of the strip, as well as along the face of the crack itself:
\begin{eqnarray}
 {\sigma_{yz}^{(1)}(x,\varepsilon H_1)=0,\qquad}{\sigma_{yz}^{(2)}(x,-\varepsilon H_2)=0,}\quad{x\in(-a/2,a/2),}\label{topbc}
 \\{\sigma_{yz}^{(1)}(x,0^+)=0,\qquad}{\sigma_{yz}^{(2)}(x,0^-)=0,}\quad{x\in(-l/2,l/2)}.
\end{eqnarray}

Outside the crack, along the boundary between $\Pi_\varepsilon^{(1)}$ and $\Pi_\varepsilon^{(2)}$, there is an imperfect interface described by the condition
\begin{equation}\label{imperfectbc}
 u^{(1)}(x,0^+)-u^{(2)}(x,0^-)=\varepsilon\kappa\sigma_{yz}^{(1)}(x,0^+),\quad x\in(-a/2,-l/2)\cup(l/2,a/2).
\end{equation}
We also assume continuity of stress across the interface
\begin{equation}\label{middlebc}
  \sigma_{yz}^{(1)}(x,0^+)=\sigma_{yz}^{(2)}(x,0^-),\quad x\in(-a/2,-l/2)\cup(l/2,a/2).
\end{equation}
We seek the solutions $u^{(j)}$ which represent the Bloch-Floquet waves, so that at the ends of our elementary cell $x=\pm a/2$ we have for $j=1,2$ the Bloch-Floquet conditions
\begin{eqnarray}
 u^{(j)}(-a/2,y)&=&e^{-iKa}u^{(j)}(a/2,y),\qquad y\in(-\varepsilon H_2,\varepsilon H_1),
  \\\sigma_{xz}^{(j)}(-a/2,y)&=&e^{-iKa}\sigma_{xz}^{(j)}(a/2,y),\qquad y\in(-\varepsilon H_2,\varepsilon H_1).
\end{eqnarray}

For a fixed value of the Bloch parameter $K$, we seek the eigenvalues $\omega$ and the corresponding eigenfunctions $u^{(j)}$ with finite norm in $W_2^1(\Pi_\varepsilon^{(j)})$, $j=1,2$.

In (\ref{imperfectbc}), the case in which $\kappa=0$ corresponds to an ideal/perfect interface between the different materials; such a problem was considered in \cite{Mish2007}. Where possible we will follow the same line as in this paper. To summarise the approach, we approximate $u$ in a certain form, derive a lower-dimensional model together with boundary layers in the vicinity of the vertices of the crack and then use our weight function to assist in the derivation of junction conditions for a skeleton model.

\subsection{Asymptotic Ansatz}
The eigenfunctions $u(x,y)$ are approximated in the form
\begin{eqnarray}\label{uansatz}
u(x,y,\varepsilon)&=&\sum\limits_{k=0}^N{\varepsilon^k}
\left\{
  \sum\limits_{m=1}^4\chi_m
  \left(
    v_m^{(k)}(x)+\varepsilon^2 V_m^{(k)}(x,Y)
  \right)\right.\nonumber\\
&+&\left.\left(
W_A^{(k)}(X_A,Y)+W_B^{(k)}(X_B,Y)
\right)
\right\}
+R_N(x,y,\varepsilon),
\end{eqnarray}
with scaled co-ordinates $X_A,$ $X_B$ and $Y$ introduced in the vicinity of the left and right vertices of the crack defined as
\begin{equation}\label{scaledcoords}
 X_A=\frac{x-x_A}{\varepsilon},\qquad X_B=\frac{x-x_B}{\varepsilon},\qquad Y=\frac{y}{\varepsilon}.
\end{equation}
Here, $v_k^{(m)}$ represent solutions of lower-dimensional problems within limit sets $\Omega_0^{(j)}$, $j=1,2,3,4$. $\chi_m=\chi_m(x,y,\varepsilon)$ are cut-off functions defined so that $\chi_m(x,y;\varepsilon)\equiv1$ in $\Omega_\varepsilon^{(m)}$ and decay rapidly to zero outside $\Omega_\varepsilon^{(m)}$. They vanish near the so-called junction points $A$ and $B$ (the vertices of the crack). The terms $W_A^{(k)}$ and $W_B^{(k)}$ represent the boundary layers near $A$ and $B$, and  $V_m^{(k)}$ is the `fast' change of eigenfunctions in the transverse direction in the domain $\Omega_\varepsilon^{(j)}$. $R_N$ is the remainder term in the asymptotic approximation. We would like to indicate to the reader that the uppercase scaled co-ordinate $X_B$ defined in (\ref{scaledcoords}) corresponds to $X$ from the derivation of the weight function in section \ref{weightsection}.

We note that this form of Ansatz relies upon the vital assumption that the boundary layers surrounding the crack vertices $A$ and $B$ are independent. That is, we assume that the exponential decay of both boundary layers is sufficiently rapid so that it is negligible in the vicinity of the other crack tip.

In this paper we will consider the form of approximation given in (\ref{uansatz}) with $N=1$ and will comment on the effect of taking higher order approximations.

\subsection{One-dimensional model problems}
Outside the vicinity of $A$ and $B$, the boundary layers $W_A^{(j)}$ and $W_B^{(j)}$ decay (we later verify this to be the case) and so seek $u$ in the form
\begin{equation}\label{notboundarylayer}
 u(x,y,\varepsilon)\sim\sum\limits_{k=0}^1{\varepsilon^k\left(v_m^{(k)}(x)+\varepsilon^2 V_m^{(k)}(x,Y)\right)},
\end{equation}
where $V_m^{(k)}$ have zero average over the cross-section of $\Omega_\varepsilon^{(m)}$ for all $m=1,2,3,4$. That is,
\begin{equation}\label{zerocs}
 \int_0^{H_1}{V_m^{(k)}}(x,Y)dY=0,\qquad\int\limits_{-H_2}^0{V_m^{(k)}}(x,Y)dY=0.
\end{equation}
Since the low-dimensional model problem studied in \cite{Mish2007} was the same above and below the crack (in $\Omega_\varepsilon^{(2)}$ and $\Omega_\varepsilon^{(3)}$), we refer the reader to that paper. The problem is however differently formulated in $\Omega_\varepsilon^{(1)}$ and $\Omega_\varepsilon^{(4)}$ due to the imperfect transmission conditions in these domains.
We focus our attention on the layered structure $\Omega_\varepsilon^{(1)}$; analogous arguments will apply to $\Omega_\varepsilon^{(4)}$.
We use the notation $v_{1j}^{(k)}$ to denote the function $v_{1}^{(k)}$ in $\Pi_\varepsilon^{(j)}$. The key observation is then to note that the transmission condition across the imperfect interface as given in (\ref{imperfectbc}) imply that
\begin{equation}\label{kiszero}
 v_{11}^{(k)}-v_{12}^{(k)}=0,\quad k=0,1.
\end{equation}
and so it follows that for $k=0,1,$ that the solution to this low dimensional model is not impacted by the presence of the imperfect interface.

To conclude this section, we have found that our case with the imperfect interface has the same equations for the low dimensional model up to terms in $\varepsilon$ as the case with the perfect interface studied in \cite{Mish2007}. The equations for $v_4^{(k)}$ and $V_4^{(k)}$ are of course similar to the case examined here where $m=1$. We would like to stress that the imperfect interface impacts on the low dimensional model equations for terms in $\varepsilon^k$, $k\geq2$. The equations gained in this section need to be complemented with the boundary conditions and junction conditions at the points $x_A$ and $x_B$. In order to derive these junction conditions which depend on the imperfect parameter $\kappa$, we construct boundary layers in the vicinity of the vertices of the crack.

\section{Junction conditions}
We introduce four smooth cut-off functions $\chi_m\in C^\infty({\mathbb R})$ in the spirit of \cite{Mish2007}. These are functions defined so that $\chi_m(x,y;\varepsilon)\equiv1$ in $\Omega_\varepsilon^{(m)}$ and decay rapidly to zero outside $\Omega_\varepsilon^{(m)}$.
These allow us to extend the function (\ref{notboundarylayer}) outside $\Omega_\varepsilon^{(m)}$, $m=1,2,3,4$, giving
\begin{equation}
 u(x,y;\varepsilon)\sim\sum_{k=0}^1\varepsilon^k\sum_{m=1}^4{\chi_m}(x,y,\varepsilon)\left(v_m^{(k)}(x)+\varepsilon^2 V_m^{(k)}(x,Y)\right),
\end{equation}
however this gives an error near the junction points $x_A$ and $x_B$. We therefore introduce boundary layers $W_A(X_A,Y)$ and $W_B(X_B,Y)$, and so seek $u(x,y,\varepsilon)$ in the form
\[
 \hspace{2ex}u\sim\sum_{k=0}^1\varepsilon^k\left\{\sum_{m=1}^4{\chi_m}\left(v_m^{(k)}(x)+\varepsilon^2 V_m^{(k)}(x,Y)\right)+W_A^{(k)}(X_A,Y)+W_B^{(k)}(X_B,Y)\right\}.
\]
Substituting this expression into the original equation and comparing terms of the same degree of $\varepsilon$ we obtain
\begin{equation}
 \nabla^2_{X_\alpha Y}\left\{W_\alpha^{(k)}(X_\alpha,Y)+\mathcal{F}_\alpha^{(k)}(X_\alpha,Y)\right\}=0,\quad \alpha=A,B,\quad k=0,1,
\end{equation}
with the functions $\mathcal{F}_\alpha^{(k)}$, $k=0,1$, $\alpha=A,B$ given by
\[
 \mathcal{F}_A^{(0)}=\sum\limits_{m=1}^{3}{v_m^{(0)}(x_A)\chi_m(x,y;\varepsilon)},\quad \mathcal{F}_B^{(0)}=\sum\limits_{m=2}^{4}{v_m^{(0)}(x_B)\chi_m(x,y;\varepsilon)},
\]
\begin{eqnarray}
\mathcal{F}_A^{(1)}&=&\sum\limits_{m=1}^{3}{\left\{(v_m^{(0)})'(x_A)X_A+v_m^{(1)}(x_A)\right\}\chi_m(x,y;\varepsilon)},
\\\mathcal{F}_B^{(1)}&=&\sum\limits_{m=2}^{4}{\left\{(v_m^{(0)})'(x_B)X_B+v_m^{(1)}(x_B)\right\}\chi_m(x,y;\varepsilon)}.
\end{eqnarray}
We now focus our attention near $x_B$; analogous arguments apply to $x_A$. We will consider in the following analysis four functions $g_i$, $i=1,2,3,4,$ which are solutions of the Laplace equation. These solutions also satisfy the boundary conditions corresponding to zero stress on the top and bottom edges of the strip (\ref{boundcons1}) as well as along the cut itself (\ref{boundcons3}).
They also satisfy the transmission condition (\ref{transcons1}) across the imperfect interface, along with continuity of stress (\ref{transcons2}).
These solutions are given by
\begin{equation}
 g_1=1,\qquad g_2=X_B,\qquad g_3=\mathcal{Y},\qquad g_4=\frac{\partial\mathcal{Y}}{\partial X},
\end{equation}
where $\mathcal{Y}$ is the weight function derived in section \ref{weightsection}.

Since they are boundary layers, we expect that $W_B^{(k)}$ decay exponentially as $X\to+\infty$ and behave as $C_j^{(k)}X+D_j^{(k)}$ as $X\to-\infty$. We first express $C_j^{(k)}$, $D_j^{(k)}$, $k=0,1$ in terms of $v_m^{(k)}$ and their derivatives. We have from Green's formula that
\begin{equation}\label{greens}
 0=\sum\limits_{j=1}^2{\mu_j}\int\limits_{\partial \Pi_B^{(j)}(L)}\left(g_i\frac{\partial}{\partial n}(W_B^{(k)}+\mathcal{F}_B^{(k)})-(W_B^{(k)}+\mathcal{F}_B^{(k)})\frac{\partial g_i}{\partial n}\right)dS.
\end{equation}
The further analysis is quite similar to that in \cite{Mish2007}, although we would like to stress that the weight function $\mathcal{Y}$ in the present paper is different, as are the transmission conditions. We therefore need to prepare this analysis from the beginning where it is different for $g_3,$ $g_4$.
\begin{figure}[t]
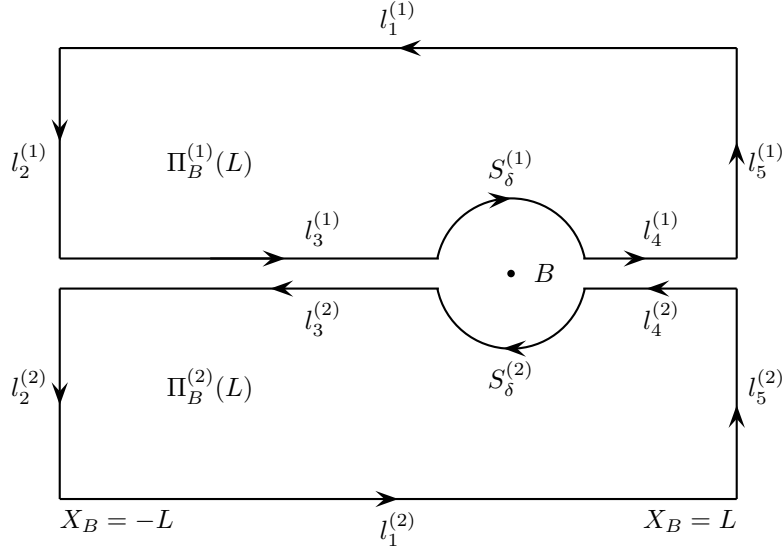

\begin{center}
\include{contourofintegration}
\end{center}
\caption{Contour of integration for (\ref{greens})}\label{contour}
\end{figure}
\subsection{The cases $k=0,1$, $i=1,2,3$}
We see from boundary conditions that integrals over the horizontal parts of the boundary $l_1^{(j)},l_3^{(j)},l_4^{(j)}$, $j=1,2$ give zero contribution to the integral. Moreover, the contribution from $S_\delta$ also disappears as $\delta\to0$ (see Figure \ref{contour}) for $g_1$, $g_2$ and $g_3$, leaving contributions solely from $l_2^{(j)}$ and $l_5^{(j)}$ in these cases.

From the definitions of $\mathcal{F}_B^{(k)}$, we obtain the following limits at $\pm\infty$ for $k=0,1$:
\begin{eqnarray}
 \mathcal{F}_B^{(0)}&=&v_4^{(0)}(x_B),\quad X_B\to+\infty,
\\\mathcal{F}_B^{(0)}&=&v_2^{(0)}(x_B)\mathcal{H}(Y)+v_3^{(0)}(x_B)\mathcal{H}(-Y),\quad X_B\to-\infty,
\\\mathcal{F}_B^{(1)}&=&(v_4^{(0)})'(x_B)X_B+v_4^{(1)}(x_B),\quad X_B\to+\infty,
\\\mathcal{F}_B^{(1)}&=&\sum\limits_{j=2}^3\left\{(v_j^{(0)})'(x_B)X_B+v_j^{(1)}(x_B)\right\}\mathcal{H}((-1)^{j}Y),\quad  X_B\to-\infty,
\end{eqnarray}
where ${\mathcal{H}}(Y)$ is the Heaviside step function.
Since $W_B^{(k)}\to0$ as $X_B\to+\infty$, equation (\ref{greens}) reduces to
\begin{eqnarray}\label{l2l5}
 0&=&\sum\limits_{j=1}^2\mu_j
\int\limits_{l_5^{(j)}}{\left(g_i\frac{\partial}{\partial X_B}\mathcal{F}_B^{(k)}-\mathcal{F}_B^{(k)}\frac{\partial g_i}{\partial X_B}\right)dS}\\&-&\sum\limits_{j=1}^2\mu_j\int\limits_{l_2^{(j)}}{\left(g_i\frac{\partial}{\partial X_B}\left(\mathcal{F}_B^{(k)}+W_B^{(k)}\right)-\left(\mathcal{F}_B^{(k)}+W_B^{(k)}\right)\frac{\partial g_i}{\partial X_B}\right)dS}
\end{eqnarray}
Applying this procedure with each of $g_1,g_2,g_3$ and $\mathcal{F}_B^{(0)},\mathcal{F}_B^{(1)}$ yields six equations, which are presented in subsection \ref{junctionsub}.

\subsection{The cases $k=0,1$, $i=4$}
To obtain a further two equations, we apply the same procedure to the solution $g_4 = \frac{\partial\mathcal{Y}}{\partial X_B}$.
Again, the contribution from the horizontal parts of the contour of integration is zero, leaving nonzero contributions from the vertical parts of the contour, $l_2^{(j)}$ and $l_5^{(j)}$. Unlike with $g_1$, $g_2$ and $g_3$ however, the contribution from ${S_\delta}^{(j)}$ is non-zero. We investigate the behaviour of $g_4$ near the crack tip.

We have that $g_4^{(j)}=\frac{\partial \mathcal{Y}_j}{\partial X}=\frac{\partial \mathcal{Y}_j}{\partial R}\cos\theta-\frac{1}{R}\frac{\partial \mathcal{Y}_j}{\partial\theta}\sin\theta,$ where $(R,\theta)$ is the usual polar co-ordinate system, with $R=\sqrt{X_B^2+Y^2}$ and so from our asymptotic estimate for $\mathcal{Y}_j$ near the crack tip from we deduce that near the crack tip,
\begin{equation}\label{cracktipfory}
 g_4^{(j)}\sim\frac{(-1)^j}{\pi\mu_j}\left\{b_0^{(\mathcal{Y})}+a_0^{(\mathcal{Y})}\ln R+(-1)^{(j+1)}a_0^{(\mathcal{Y})}\sin2\theta(\pi+(-1)^j\theta)\right\},
\end{equation}
and so for small $R$,
\begin{equation}
 \frac{\partial g_4^{(j)}}{\partial R}\sim\frac{(-1)^ja_0^{(\mathcal{Y})}}{\pi\mu_jR}.
\end{equation}
Noting that the outward normal to $S_\delta^{(j)}$ is in the direction of $-R$, we have that as $\delta\to0$
\begin{eqnarray}
\mu_j\int\limits_{S_\delta}\left(g_4\left(-\frac{\partial}{\partial R}\right)(W_B^{(k)}+\mathcal{F}_B^{(k)})-(W_B^{(k)}+\mathcal{F}_B^{(k)})\left(-\frac{\partial g_4}{\partial R}\right)\right)dS
\\=\mu_j\int\limits_{S_\delta}\left((W_B^{(k)}+\mathcal{F}_B^{(k)})\frac{\partial g_4}{\partial R}\right)Rd\theta
=\mu_j\int\limits_{S_\delta}\frac{(-1)^ja_0^{(\mathcal{Y})}}{\pi\mu_jR}\left(W_B^{(k)}+\mathcal{F}_B^{(k)}\right)Rd\theta.\nonumber
\end{eqnarray}
Since $W_B$ satisfies the same model problem as $\mathcal{Y}$, it too will possess asymptotic behaviour at the crack tip of the same form as $g_4$ in (\ref{cracktipfory}), but with different constants which we denote $a_{(k)}^{(W)}$ and $b_{(k)}^{(W)}$ for $k=0,1$. The contribution to the integral from the circular part of the contour is therefore given by
\begin{eqnarray}
-\frac{a_0^{(\mathcal{Y})}}{\pi}\int\limits_0^\pi\left(W_B^{(k)}(0^+,\theta)+\mathcal{F}_B^{(k)}(0^+,\theta)\right)d\theta
+\frac{a_0^{(\mathcal{Y})}}{\pi}\int\limits_{-\pi}^0\left(W_B^{(k)}(0^+,\theta)+\mathcal{F}_B^{(k)}(0^+,\theta)\right)d\theta \nonumber
\\=-\frac{a_0^{(\mathcal{Y})}}{\pi}\int\limits_0^\pi\frac{-1}{\pi\mu_1}\frac{\mu_1\kappa\pi}{1+\frac{\mu_1}{\mu_2}}a_{(k)}^{(W)}d\theta+
\frac{a_0^{(\mathcal{Y})}}{\pi}\int\limits_{-\pi}^0\frac{1}{\pi\mu_2}\frac{\mu_1\kappa\pi}{1+\frac{\mu_1}{\mu_2}}a_{(k)}^{(W)}d\theta
=\kappa a_0^{(\mathcal{Y})}a_{(k)}^{(W)}.\nonumber
\end{eqnarray}
With this information at hand, we are now able to apply (\ref{greens}) with $g_4$ and $\mathcal{F}_B^{(1)}, \mathcal{F}_B^{(2)}$, yielding a further two relationships.

\subsection{Deriving the junction conditions}\label{junctionsub}
We define the column matrices
\begin{equation}E^{(k)}=\left[
 \begin{array}{cccc}
  C_1^{{(k)}}&C_2^{{(k)}}&D_1^{{(k)}}&D_2^{{(k)}}
 \end{array}\right]^T,\qquad k=0,1.
\end{equation}
The eight equations obtained in the previous two subsections can then be rewritten as two matrix equations, the first of which is found to be
\begin{equation}\label{zeromatrix}
  ME^{(0)}
=\left[
\begin{array}{c}
   0\\{(\mu_1H_1+\mu_2H_2)v_4^{(0)}(x_B)-\mu_1H_1v_2^{(0)}(x_B)-\mu_2H_2v_3^{(0)}(x_B)}\\{\mu_1H_1C_1v_2^{(0)}(x_B)+\mu_2H_2C_2v_3^{(0)}(x_B)}\\{\kappa a_0^{(\mathcal{Y})}a_0^{(W)}},
\end{array}
\right]
\end{equation}
where $M$ is the 4x4 matrix
\begin{equation}\label{mmatrix}
\left[
\begin{array}{cccc}
  \mu_1H_1&\mu_2H_2&0&0\\
  0&0&\mu_1H_1&\mu_2H_2\\
  \mu_1H_1D_1&\mu_2H_2D_2&-\mu_1H_1C_1&-\mu_2H_2C_2\\
  \\\mu_1H_1C_1&\mu_2H_2C_2&0&0
 \end{array},
\right]
\end{equation}
where $C_j$ and $D_j$ are the asymptotic constants from the weight function defined in (\ref{cjdj}). The determinant of $M$ is given by $\det(M)=-\mu_1^2\mu_2^2H_1^2H_2^2(C_1-C_2)^2<0$. Therefore for $C_1^{(0)}=C_2^{(0)}=D_1^{(0)}=D_2^{(0)}=0$ (that is, for $W$ to vanish far away from the crack tip as we would expect for such a boundary layer), we have that the matrix in the right hand side of (\ref{zeromatrix}) must be equal to zero. From this follow the junction conditions
\begin{eqnarray}
 v_2^{(0)}(x_B)&=&v_3^{(0)}(x_B)=v_4^{(0)}(x_B),\label{zerojc1}
\\ a_{(0)}^{(W)}&=&0.\label{zerojc2}
\end{eqnarray}
The latter condition (\ref{zerojc2}) yields that $W_B^{(0)}\equiv0$.
The second matrix equation is
\begin{equation}\label{onematrix}
ME^{(1)}=
\end{equation}
\[
\left[\begin{array}{c}
   {(\mu_1H_1+\mu_2H_2)(v_4^{(0)})'(x_B)-\mu_1H_1(v_2^{(0)})'(x_B)-\mu_2H_2(v_3^{(0)})'(x_B)}\\{(\mu_1H_1+\mu_2H_2)v_4^{(1)}(x_B)-\mu_1H_1v_2^{(1)}(x_B)-\mu_2H_2v_3^{(1)}(x_B)}\\
  {\mu_1H_1C_1v_2^{(1)}(x_B)+\mu_2H_2C_2v_3^{(1)}(x_B)-\mu_1H_1D_1(v_2^{(0)})'(x_B)-\mu_2H_2D_2(v_3^{(0)})'(x_B)}\\
  {\kappa a_0^{(\mathcal{Y})}a_1^{(W)}-\mu_1H_1C_1(v_2^{(0)})'(x_B)-\mu_2H_2C_2(v_3^{(0)})'(x_B)}
  \end{array}\right]
\]
where $M$ is the matrix given in (\ref{mmatrix}). For $C_1^{(1)}=C_2^{(1)}=D_1^{(1)}=D_2^{(1)}=0$, the right hand matrix is again set to zero. Noting that $a_0^{(\mathcal{Y})}=\lambda$ (see (\ref{a0value})) and that $\mu_1H_1C_1+\mu_2H_2C_2=0$, setting the fourth row of the RHS matrix to zero then yields that
\begin{equation}
a_{(1)}^{(W)}=\frac{1}{\kappa\lambda}\Delta\{(v^{(0)})'\}.\label{firstjc3}
\end{equation}
where
\begin{equation}
\Delta\{(v^{(0)})'\}(x_B)= (v_2^{(0)})'(x_B)-(v_3^{(0)})'(x_B).
\end{equation}
The other conditions imply
\begin{eqnarray}
 v_2^{(1)}(x_B)&=v_4^{(1)}(x_B)-\frac{\mu_2H_2}{\mu_1H_1+\mu_2H_2}\left(\frac{\alpha}{\pi}+\frac{1}{\lambda}\right)\Delta\{(v^{(0)})'\}(x_B),\label{firstjc1}
\\v_3^{(1)}(x_B)&=v_4^{(1)}(x_B)+\frac{\mu_1H_1}{\mu_1H_1+\mu_2H_2}\left(\frac{\alpha}{\pi}+\frac{1}{\lambda}\right)\Delta\{(v^{(0)})'\}(x_B),\label{firstjc2}
\end{eqnarray}
along with the relationship
\begin{equation}\label{firstjc4}
 (\mu_1H_1+\mu_2H_2)(v_4^{(0)})'(x_B)-\mu_1H_1(v_2^{(0)})'(x_B)-\mu_2H_2(v_3^{(0)})'(x_B)=0.
\end{equation}
\textcolor{black}{We stress that $\alpha$ and $\lambda$ are functions of $\kappa$ and so expressions (\ref{firstjc1}) and (\ref{firstjc2}) describe how the junction conditions depend upon the extent of imperfection of the interface. In particular, $(\alpha/\pi+1/\lambda)$ is a constant that plays a crucial physical role since it defines the proportionality between the displacement jump in the first order approximation and the angle of opening in the zero order approximation. Equation (\ref{firstjc4}) complements conditions (\ref{zerojc1}) and (\ref{zerojc2}) to give full information for the zero order approximation. We later present numerical results for the normalized constant $\alpha_I=(\alpha/\pi)+1/\lambda)/(H_1+H_2)$.}

The conditions regarding the first order approximation (\ref{firstjc3}), (\ref{firstjc1}) and (\ref{firstjc2}) can be complemented by a further equation in $(v_m^{(1)})'(x_B)$, which follows from the next level of approximation, i.e. taking $N=2$ in (\ref{uansatz}):
\[
 \mu_1H_1(v_2^{(1)})'(x_B)+\mu_2H_2(v_3^{(1)})'(x_B)-(\mu_1H_1+\mu_2H_2)(v_4^{(1)})'(x_B)=\sum\limits_{j=1}^2\int\limits_{\Pi_B^{(j)}}\frac{\omega^2}{c_j^2}W_B^{(0)}\mathrm{d}\Pi_B^{(j)},
\]
and by our earlier comment that $W_B^{(0)}\equiv0$, the right side of this expression is zero.
\textcolor{black}{At this point we would like to comment that taking higher order approximations and evaluating higher order junction conditions is possible but much more advanced. For example, integrals analogous to that on the right hand side of the above expression would depend upon $W_B^{(1)}$ and boundary layers from higher order approximations, and so would not in general be zero. However, since we focus on thin strips, $\varepsilon$ is small and so terms in $\varepsilon^2$ would give significantly less contribution than the lower order approximations. We later comment on the accuracy of the zero order approximation on page \pageref{fem_comparison_comments} by comparing computations against FEM results in a case where $\varepsilon$ is not too small. We would like to underline that the accuracy will increase for smaller $\varepsilon$, but for very small $\varepsilon$ it is no longer possible to obtain finite element computations.}

\section{Numerical simulations and discussions}\label{section:comp}
\begin{figure}[t]
  \includegraphics[width=0.9\linewidth]{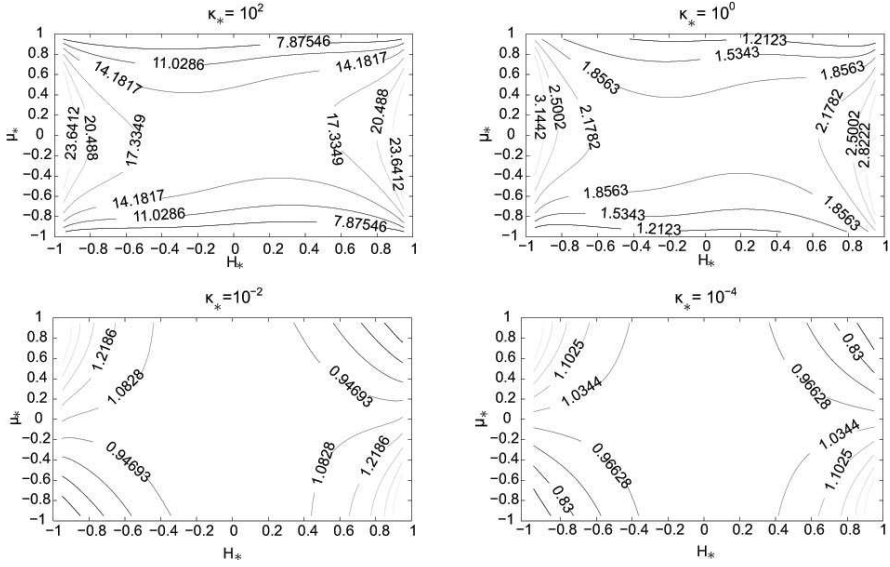}
  \caption{{\footnotesize{Contour plots of the ratio $\alpha_I/\alpha_P$ for four different values of $\kappa_*$, a dimensionless parameter describing the extent of imperfection of the interface between the two materials. The axes of each plot are $\mu_*$ and $H_*$, dimensionless parameters respectively describing the mechanical and geometric properties of the problem. The ratio $\alpha_I/\alpha_P$ gets closer to $1$ as $\kappa_*$ decreases in value towards $0$.}}}
  \label{contourplot}
\end{figure}

To enable us to compare results with the perfect interface case discussed in \cite{Mish2007} effectively, we seek normalized constants. We first seek a normalized representation of $\alpha$. We introduce the notation
\[
 H=H_1+H_2,\; H_*=\frac{H_1-H_2}{H_1+H_2},\; \mu_*=\frac{\mu_1-\mu_2}{\mu_1+\mu_2},\; \kappa_*=\frac{\kappa(\mu_1+\mu_2)}{H},\;\lambda_*=\lambda H,
\]
where $H_*$, $\mu_*$ and $\kappa_*$ are non-dimensional parameters which respectively describe the geometrical, mechanical and imperfect properties of the problem. $\lambda$ is the constant dependent on $\mu_j$, $H_j$ and $\kappa$ defined in (\ref{lambdadef}). $\lambda_*$ can be expressed in terms of the other dimensionless parameters as
\begin{equation}
 \lambda_*^2=\frac{8(1+\mu_*H_*)}{\kappa_*(1-\mu_*^2)(1-H_*^2)}.
\end{equation}
 We also introduce the function
\[
 \Xi_{**}(t)=\frac{t}{\lambda_*^2+t^2}\left(t+\frac{2}{\kappa_*(1+\mu_*)}\coth\frac{t(1+H_*)}{2}+\frac{2}{\kappa_*(1-\mu_*)}\coth\frac{t(1-H_*)}{2}\right),
\]
which satisfies the relationship $\Xi_{**}(t)=\Xi_*\left(\frac{t}{H}\right)$, and so we can write
\begin{equation}
 \alpha=\int\limits_0^\infty\frac{\ln\Xi_*(\xi)}{\xi^2}d\xi=\int\limits_0^\infty{\frac{H^2\ln\Xi_{**}(t)}{t^2}\frac{dt}{H}}=H\int\limits_0^\infty\frac{\ln\Xi_{**}(t)}{t^2}dt=H\alpha_*,
\end{equation}
where we have defined the non-dimensional quantity $\alpha_*$. We find through asymptotic analysis that
\begin{equation}
 \frac{\ln\Xi_{**}(t)}{t^2}=\frac{1}{12}\frac{H_*^3\mu_*-H_*^2-\mu_*H_*+1}{1+\mu_* H_*}+O(t^2),\quad t\to0.
\end{equation}

Mishuris, Movchan and Bercial \cite{Mish2007} showed that in the analogous problem to that discussed in this paper with a perfect interface instead of an imperfect interface,
\begin{equation}
 D_j=\alpha_{{P}}(H_1+H_2)C_j,
\end{equation}
where
\[
 \alpha_P=\frac{1}{\pi}\ln\left\{\left(\frac{1+H_*}{2}\right)^{\frac{1+H_*}{2}}\left(\frac{1-H_*}{2}\right)^{\frac{1-H_*}{2}}\right\}-\frac{\mu_*}{\pi}\int\limits_0^\infty\frac{H_*-\tanh(tH_*)\coth(t)}{(\sinh(t)+\mu_*\sinh(tH_*))t}dt.
\]
We have demonstrated (see the form of the constants $C_j$, $D_j$ in (\ref{cjdj})) that for the imperfect interface problem,
\begin{equation}
 D_j=\alpha_{{I}}(H_1+H_2)C_j,
\end{equation}
where
\begin{equation}
 \alpha_I=-\left(\frac{1}{\pi}\int\limits_0^\infty\frac{\ln\Xi_{**}(t)}{t^2}dt+\frac{1}{\lambda_*}\right),
\end{equation}
and since small $\kappa_*$ correspond to an interface which is `almost perfect', we would expect $\alpha_I\to\alpha_P$ as $\kappa_*\to0$.
Figure \ref{contourplot} shows a plot of the ratio $\alpha_I/\alpha_P$ on axes of $\mu_*$ against $H_*$ for four different values of $\kappa_*$. From this it is easily seen that as $\kappa_*\to0$, $\alpha_I/\alpha_P$ gets close to $1$ as expected. The behaviour of the weight functions near the crack tip are however absolutely different since the problem is singularly perturbed, that is:
\begin{equation}
 [\mathcal{Y}]\sim\sqrt{\eta\kappa},\quad\kappa\to0,
\end{equation}
\begin{equation}
 \mu_1\left.\frac{\partial\mathcal{Y}}{\partial Y}\right|_{Y=0+}\sim-\sqrt{\frac{\eta}{\kappa}},\quad X\to0,\quad\kappa\to0,
\end{equation}
where $\eta$ is defined in (\ref{estofxi0}).
In the plots showing the ratio for small values of $\kappa_*$, the highest deviations from 1 occur near the corner of the plot. These correspond to the cases where there is a high contrast between the shear moduli and thicknesses of the two materials. We see that in the cases where the materials have similar shear moduli and thicknesses (nearer the center of the plot), the ratio $\alpha_I/\alpha_P$ quickly approaches 1 as $\kappa_*\to0$.

Figure \ref{alphaplots} shows surface plots of $\alpha_I$ on axes of $\mu_*$ and $H_*$ for $\kappa_*=100,$ 1, and 0.01. This constant describes the impact that the imperfect interface has upon the junction conditions as described in equations (\ref{firstjc1}) and (\ref{firstjc2}). Also shown in the figure is a plot of $\alpha_P$. The similarity between the plot of $\alpha_I$ for $\kappa_*=0.01$ and the plot of $\alpha_P$ is evident here. For the cases with larger $\kappa_*$ values, we see that $\alpha_P$ is differently dependent upon the mechanical and geometric parameters of the problem.

\begin{figure}[t]
\begin{minipage}{0.9\linewidth}{
     \includegraphics[width=6cm]{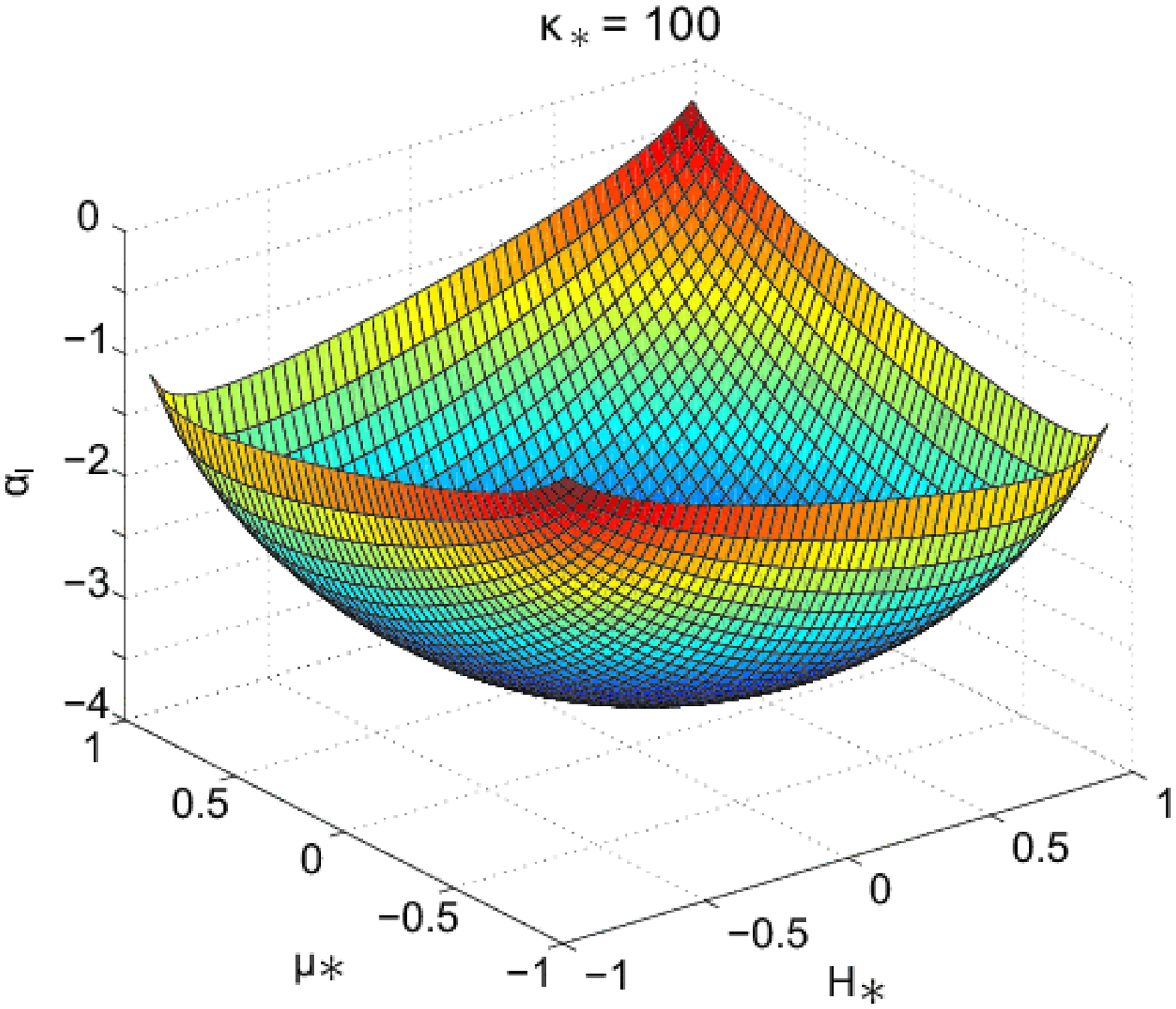}\quad
     \includegraphics[width=6cm]{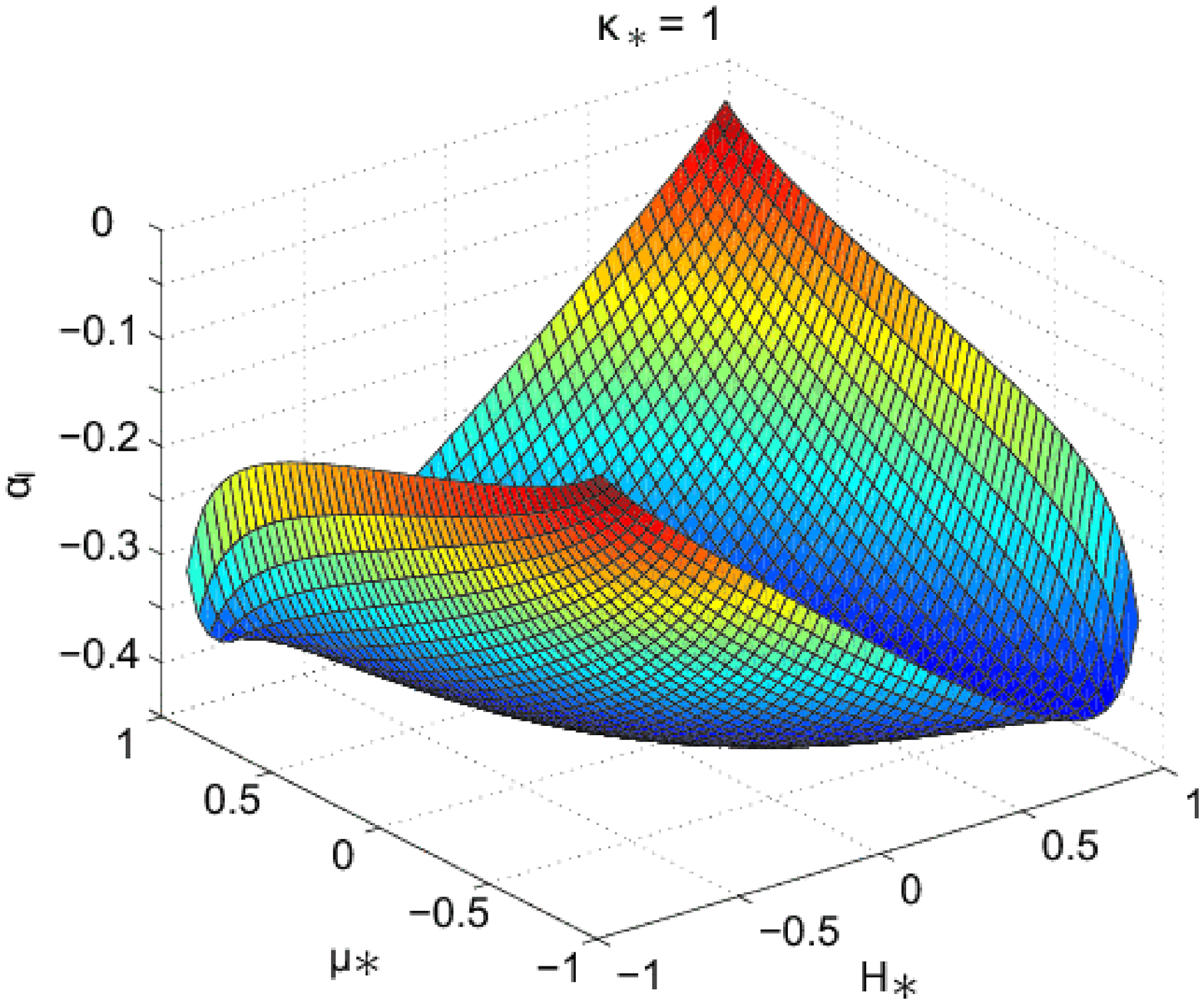}\vspace{1em}\\
     \includegraphics[width=6cm]{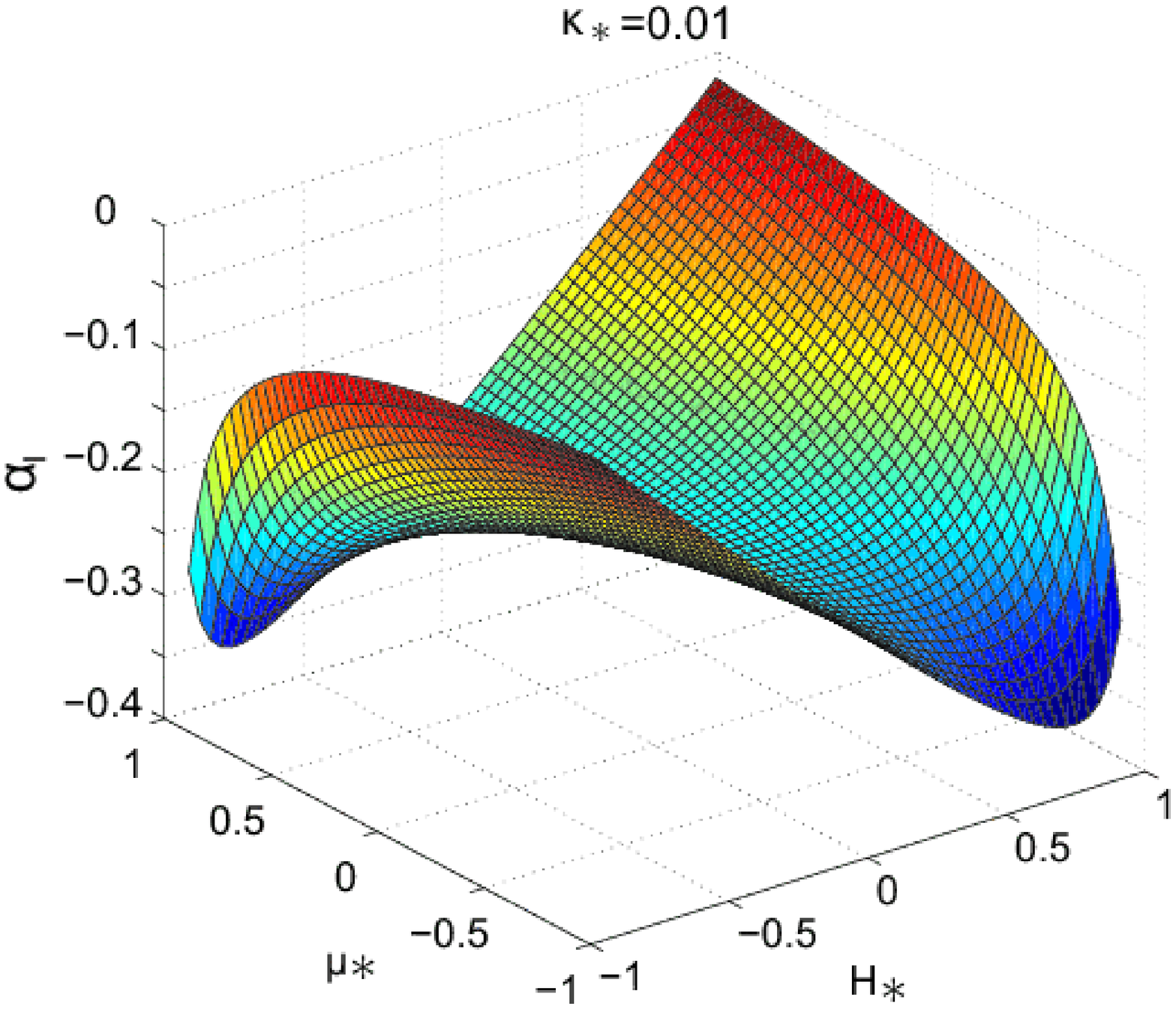}\quad
     \includegraphics[width=6cm]{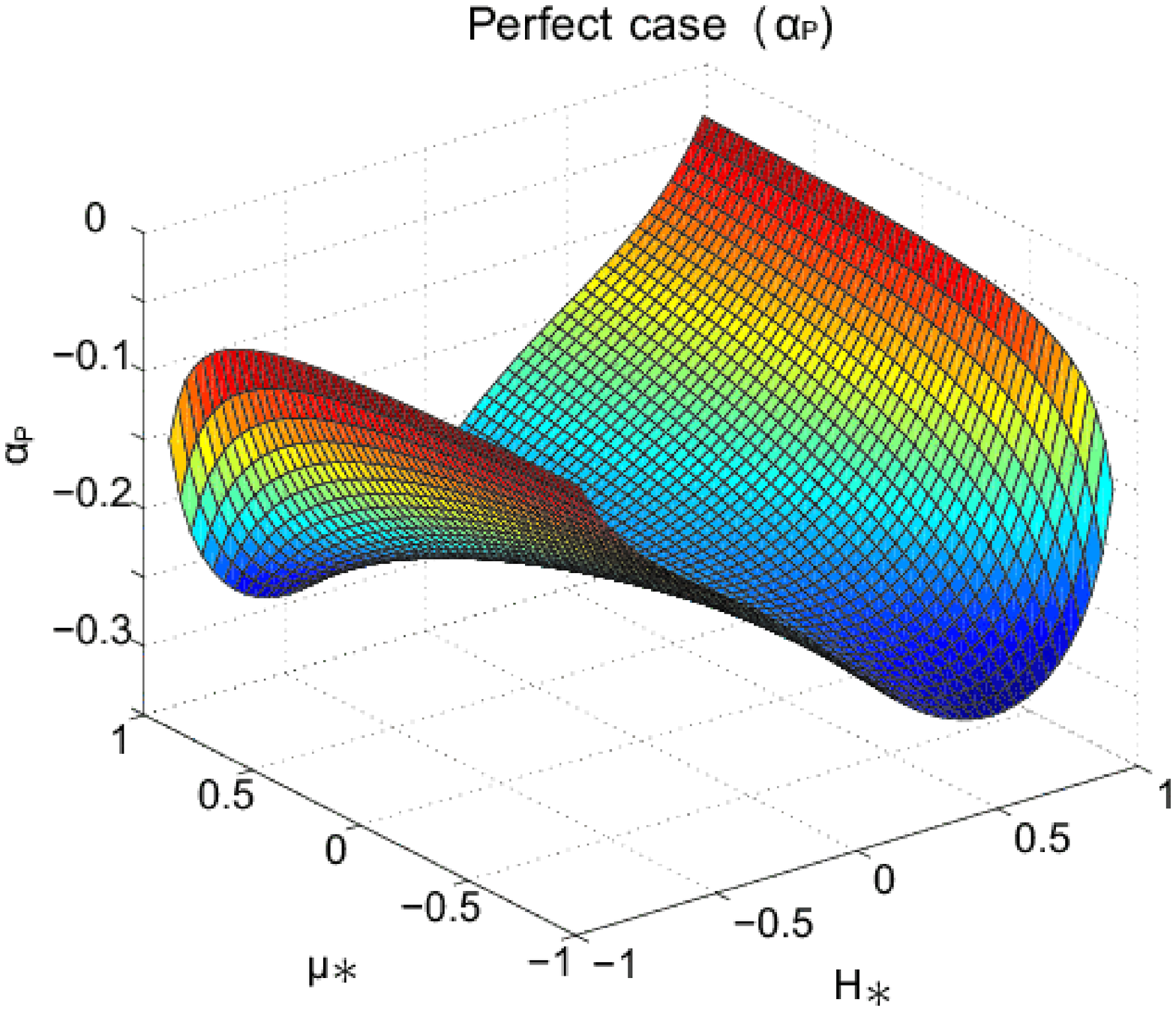}}\end{minipage}
\caption{{\footnotesize{Surface plots of $\alpha_I$ for $\kappa_*=100,$ 1 and 0.01; also of $\alpha_P$, all plotted on axes of $\mu_*$ and $H_*$.}}}
\label{alphaplots}
 \end{figure}

Figure \ref{fig:comsol} shows finite-element plots (COMSOL) of standing wave eigensolutions. For these simulations we use the following geometrical parameters for the elementary cell:
\[
 l=0.8[\mathrm{m}],\quad a=2.4[\mathrm{m}],\quad H_1=0.1[\mathrm{m}],\quad H_2=0.05[\mathrm{m}],\nonumber
\]
and the following material constants which correspond to iron (in $\Pi_\varepsilon^{(2)}$) and aluminium (in $\Pi_\varepsilon^{(1)}$).
\[
\mu_2=82\cdot10^9[\mathrm{N/m^2}],\quad \mu_1=26\cdot10^9[\mathrm{N/m^2}],\quad\rho_2=7860[\mathrm{kg/m^3}],\quad\rho_1=2700[\mathrm{kg/m^3}].
\]
Presented in this figure are three plots corresponding to Al-Fe strips with different materials bonding them together, with the vertical dotted lines indicating the location of the crack tips. The imperfect interface is modelled in the COMSOL simulations by a thin layer occupied by an adhesive material; this approach was justified in \cite{Mish2001a,Mish2006}, among others. Provided that $h_{\mathrm{resin}}/H_2$ is sufficiently small and $\mu_{\mathrm{resin}}$ is small in comparison to $\mu_1$ and $\mu_2$, this gives $\kappa=h_{\mathrm{resin}}/\mu_{\mathrm{resin}}$.

The second of the three plots in Figure \ref{fig:comsol} uses epoxy resin as the bonding material with parameters
\[
 \mu_{\mathrm{resin}}=2.5\cdot10^9[\mathrm{N/m^2}],\quad \rho_{\mathrm{resin}}=1850[\mathrm{kg/m^3}],\quad h_{\mathrm{resin}}=0.01[\mathrm{m}].
\]
For comparison, the first plot shows a simulation with a gluing layer of shear modulus 1000 greater than that of epoxy resin. The third plot uses a material with shear modulus 10 times less than epoxy resin. Equivalently, these three cases in the top, middle and bottom parts of the figure correspond to $\kappa_*=2.88\cdot10^{-3},$ $\kappa_*=2.88,$ and $\kappa_*=28.8,$ respectively. The plots show that the standing wave is more localised and intense in the locality of the crack when the bonding material is stiffer. Conversely, when the bonding material is less stiff, the standing wave extends further beyond the locality of the crack and is less intense. Closely packed contours indicate areas where stress is high; we see that the highest stress is to be found in the vicinity of the crack tip in all three cases. Moreover, as we would expect, the highest stress intensity is found in the case with the stiffest bonding material. 

\begin{figure}[t]
\begin{center}
\includegraphics[width=11cm]{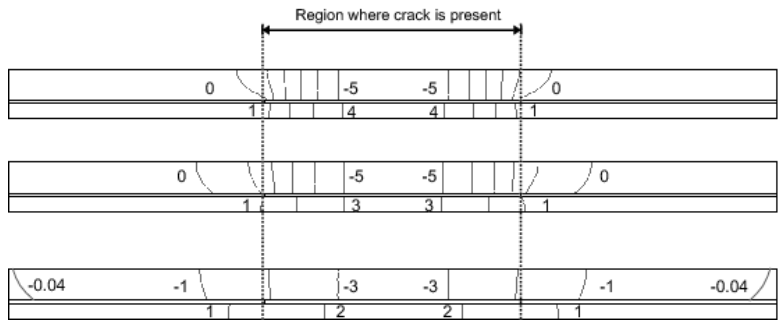}
\end{center}
\caption{\footnotesize{Finite element computation (COMSOL) contour plot of the eigensolution corresponding to the standing Bloch-Floquet waves for three different values of $\kappa$. {\bf Top: }Bonding material with shear modulus $1000\mu_\mathrm{resin}$. {\bf Middle: }Bonding material is epoxy resin. {\bf Bottom: }Bonding material with shear modulus $\mu_\mathrm{resin}/10$. Countours join points of integer values, and the dotted vertical lines indicate the location of the crack tips.}}
\label{fig:comsol}
\end{figure}

We do not present dispersion diagrams here computed by the asymptotic analysis and COMSOL as they are similar to those given in \cite{Mish2007}. As in that paper, the biggest discrepancy between results obtained from asymptotic analysis and numerical simulations appear for the case of the standing waves. \textcolor{black}{In all other situations the accuracy is very good, with a typical discrepancy between finite element and asymptotic results of around 0.3\% in the case where the strip has the same dimensions as used throughout this section, which corresponds to $\varepsilon=0.0625$.\label{fem_comparison_comments}} We remind the reader that we use static boundary layers in the analysis. The standing waves lie in the area of rather high frequencies, which may provide one possible explanation for this phenomenon. This discrepancy needs to be eliminated and this will form part of future work.

\textcolor{black}{It is readily seen in the bottom plot of Figure \ref{fig:comsol} (which corresponds to a highly imperfect interface) that the boundary layer support extends almost to the edge of the elementary cell. This extension far away from the crack tips suggests that the boundary layers decay slowly from the crack tips and so may not be assumed independent. In this case, therefore, our analysis may become invalid due to the assumption in our asymptotic procedure that the exponentially decaying boundary layer does not influence the Bloch-Floquet conditions.} This assumption is satisfied if $\gamma_+$ is far from zero, so if $\kappa$ is not too large. More accurately, we assume $\gamma_+\gg\frac{\varepsilon}{a-l}$ (see (\ref{gammaplusest}) for large $\kappa$). If the imperfect interface is too weak and this condition is violated then the junction conditions evaluated here will no longer be accurate and other analysis should be sought.

\section*{Acknowledgements}
AV would like to thank Aberystwyth University for providing APRS funding. GM is grateful for support from the European Union Seventh Framework Programme under contract number PIAP-GA-2009-251475.

\appendix

\section{Derivation of asymptotics of $\Xi_*^+(\xi)$}
We present here the derivation of asymptotics for $\Xi_*^+(\xi)$. The results of this derivation are used in expressions (\ref{xistarplus0}) and (\ref{xistarplusinf}). We introduce the auxiliary function
\begin{equation}\label{thetadef}
\Theta_*^+(\xi)=\int\limits_{-\infty-i\beta}^{\infty-i\beta}{\frac{\ln\Xi_*(t)}{t-\xi}dt},
\end{equation}
so that $\Xi_*^+(\xi)=\exp((1/2\pi i)\Theta_*^+(\xi))$ (see (\ref{xistarpm})).
We first note that $\Theta_*^+(0)=0$ since the integrand is odd and estimate (\ref{estofxi0})$_2$ demonstrates integrability of $\Xi_*$ at the zero point, allowing us to take $\beta=0$.
Thus
\[
\Theta_*^+(\xi)=\int\limits_{-\infty-i\beta}^{\infty-i\beta}{\left[\frac{\ln\Xi_*(t)}{t-\xi}-\frac{\ln\Xi_*(t)}{t}\right]dt}=\xi\int\limits_{-\infty-i\beta}^{\infty-i\beta}{\frac{\ln\Xi_*(t)}{t(t-\xi)}dt}\to0,\quad\xi\to0,
\]
since the integral is bounded.
Also, we have that
\begin{equation}
\int\limits_{-\infty-i\beta}^{\infty-i\beta}{\frac{\ln\Xi_*(t)}{t^2}dt}=\int\limits_{-\infty}^\infty{\frac{\ln\Xi_*(t)}{t^2}dt}=2\int\limits_{0}^\infty{\frac{\ln\Xi_*(t)}{t^2}dt}=2\alpha,
\end{equation}
since the integrand is even and again by considering (\ref{estofxi0})$_2$, which indicates that we have integrability at the zero point. Here we have found that
\begin{equation}
\Theta_*^+(\xi)=2\alpha\xi+O(|\xi|^2),\quad\xi\to0.
\end{equation}
From this we obtain the following estimate for $\Xi_*^+(\xi)$ as $\xi\to0$:
\begin{equation}
\Xi_*^+(\xi)=1+\frac{\alpha\xi}{\pi i}+O(|\xi|^2),\quad\xi\to0.
\end{equation}
We now seek estimates of $\Theta_*^+(\xi)$ for $\xi\to\infty$ within the domain. To avoid problems caused by integrating along the real line, we consider $\xi\to\infty$ in such a way that Im$(\xi)\to+\infty$. Integrating (\ref{thetadef}) by parts, splitting the integral in two and manipulating the resulting expression gives

\begin{equation}
\Theta_*^+(\xi)=\int\limits_0^\infty{\ln\left(\frac{1+t/\xi}{1-t/\xi}\right)\frac{\Xi'_*(t)}{\Xi_*(t)}dt}.
\end{equation}
We introduce an arbitrary $R>0$ and split this integral at $R$ to give
\begin{equation}
\Theta_*^+(\xi)=\int\limits_0^R{\ln\left(\frac{1+t/\xi}{1-t/\xi}\right)\frac{\Xi'_*(t)}{\Xi_*(t)}dt}+\int\limits_R^\infty{\ln\left(\frac{1+t/\xi}{1-t/\xi}\right)\frac{\Xi'_*(t)}{\Xi_*(t)}dt}.
\end{equation}
We then see that
\begin{equation}
\ln\left(\frac{1+t/\xi}{1-t/\xi}\right)=2\frac{t}{\xi}+O\left(\frac{t^3}{|\xi|^3}\right),\quad\xi\to\infty,\quad0<t<R,
\end{equation}
and from (\ref{bigxiasym}) we have
\begin{equation}
\frac{\Xi'_*(t)}{\Xi_*(t)}=\left[-\frac{\mu_1+\mu_2}{\kappa\mu_1\mu_2}\right]\frac{1}{t^2}+O\left(\frac{1}{t^3}\right),\quad t\to\infty.
\end{equation}
This allows us to estimate
\[
\Theta_*^+(\xi)=\int\limits_R^\infty\left[-\frac{(\mu_1+\mu_2)}{\mu_1\mu_2\kappa}\frac{1}{t^2}+O\left(\frac{1}{t^3}\right)\right]\ln\left(\frac{\xi+t}{\xi-t}\right)dt+O\left(\frac{1}{|\xi|}\right),\quad\xi\to\infty.\label{thetaform}
\]
After integrating by parts and performing a change of variables, we find that
\begin{equation}
\int\limits_R^\infty{\frac{1}{t^2}\ln\left(\frac{\xi+t}{\xi-t}\right)dt}=-\frac{1}{\xi}\left(\ln\left|\frac{1}{\xi^2}\right|+i\arg\left(-\frac{1}{\xi^2}\right)\right)+O\left(\frac{1}{|\xi|}\right),\quad\xi\to\infty,
\end{equation}
and so from (\ref{thetaform}), we deduce that
\begin{equation}
\Theta_*^+(\xi)=\frac{2(\mu_1+\mu_2)}{\mu_1\mu_2\kappa\xi}\ln(-i\xi)+O\left(\frac{1}{|\xi|}\right),\quad\mathrm{Im}(\xi)\to+\infty.
\end{equation}
Recalling the relationship between our auxiliary function $\Theta_*^+$ and $\Xi_*^+$ as we discussed after (\ref{thetadef}), we see that
\begin{equation}
\Xi_*^+(\xi)=1+\frac{1}{\pi i}\frac{(\mu_1+\mu_2)}{\mu_1\mu_2\kappa}\frac{\ln(-i\xi)}{\xi}+O\left(\frac{1}{|\xi|}\right),\quad\mathrm{Im}(\xi)\to+\infty.
\end{equation}

\section{Theorem}
\begin{theorem}\label{theoremforphi}
Let f(x) be the function
\begin{equation}
f(x)=\frac{1}{2\pi}\int\limits_{-\infty}^{\infty}\Phi^+(t)e^{-ixt}dt.
\end{equation}
If $\Phi^+(t)$ is analytic in $\mathbb{C}^+$ and
\begin{equation}
\Phi^+(t)=a_1t^{-1}+O(t^{-(1+\delta)}),\quad t\to\infty,
\end{equation}
where $\delta>0$ is small, in the closed half-plane $\overline{\mathbb{C}}^+=\mathbb{C}^+\cup\mathbb{R}$, then $f(x)=0$ for all $x<0$ and
\begin{equation}
\lim\limits_{x\to0^+}f(x)=-ia_1.
\end{equation}
\end{theorem}

\begin{proof}
The fact that $f(x)=0$ for all $x<0$ is a direct consequence of the fact that $\Phi^+(t)$ is a `+' function. Assume now that $x>0$. From the assumptions on the behaviour of the function $\Phi^+(t)$, it follows that $\Phi^+(t)=a_1t^{-1}+R(t)$, where $tR(t)\to0$, as $t\to\infty,$ $t\in\overline{\mathbb{C}}^+$ (including $t\to\pm\infty,$ $t\in\mathbb{R}$).

We write

\begin{equation}\label{fdefn}
f(x)=\frac{1}{2\pi}\lim\limits_{a\to+\infty}\left\{\int\limits_a^\infty[\Phi^+(-t)e^{ixt}+\Phi^+(t)e^{-ixt}]dt+\int\limits_{-a}^a{\Phi^+(t)e^{-ixt}}dt\right\}.
\end{equation}
The first integral is
\begin{equation}
f_1(x,a)=\int\limits_a^\infty[\Phi^+(-t)e^{ixt}+\Phi^+(t)e^{-ixt}]dt=f_{11}(x,a)+f_{12}(x,a),
\end{equation}
where
\begin{equation}
f_{11}(x,a)=\int\limits_a^\infty{\left[-\frac{a_1}{t}e^{ixt}+\frac{a_1}{t}e^{-ixt}\right]dt}=-2ia_1\int\limits_{xa}^\infty{\frac{\sin(t)}{t}dt},
\end{equation}
and
\[
f_{12}(x,a)=\int\limits_a^\infty{\left[R(-t)e^{ixt}+R(t)e^{-ixt}\right]dt}=\int\limits_{xa}^\infty{\left[\frac{1}{x}R\left(-\frac{t}{x}\right)e^{it}+\frac{1}{x}R\left(\frac{t}{x}\right)e^{-it}\right]dt}.
\]
Taking $a=x^{-1/2}$, we have that $f_{11}(x,x^{-1/2})\to-i\pi a_1$ and $f_{12}(x,x^{-1/2})\to0$ as $x\to0^+$.
Let us denote the second integral in (\ref{fdefn}) by $f_{2}(x,a)$. Then using analyticity of $\Phi^+(t)$ in $\mathbb{C}^+$ and defining
\begin{equation}
\Gamma_a=\{t\in\mathbb{C}|t=ae^{i\theta},0<\theta<\pi\},
\end{equation}
we deduce
\begin{equation}
	f_2(x,a)=-\int\limits_{\Gamma_a}{\Phi^+(t)e^{-ixt}dt}.
\end{equation}
We write this in the form
\begin{equation}
	f_2(x,a)=f_{21}(x,a)+f_{22}(x,a),
\end{equation}
where
\begin{equation}
f_{21}(x,a)=-\int\limits_{\Gamma_a}\frac{a_1}{t}e^{-ixt}dt,\quad\mathrm{ and } f_{22}(x,a)=-\int\limits_{\Gamma_a}R(t)e^{-ixt}dt.
\end{equation}
Again taking $a=x^{-1/2}$, we obtain
\begin{equation}
f_{21}(x,x^{-1/2})=-\int\limits_{\Gamma_{x^{-1/2}}}{\frac{a_1}{t}e^{-ixt}dt}\sim-a_1\int\limits_{\Gamma_{x^{-1/2}}}{\frac{1}{t}dt}
=-i\pi a_1,\quad x\to0^+.
\end{equation}
Now,
\begin{equation}
f_{22}(x,a)=-\int\limits_{\Gamma_a}{R(t)e^{-ixt}dt}=-\int\limits_{\Gamma_{xa}}{\frac{1}{x}R\left(\frac{t}{x}\right)e^{-it}dt},
\end{equation}
and so $f_{22}(x,x^{-1/2})\to0$ as $x\to0^+$. By collecting together these observations and reconsidering equation (\ref{fdefn}) we conclude that
\begin{equation}
f(x)\to\frac{1}{2\pi}(-i\pi a_1 -i\pi a_1)=-ia_1, \quad x\to0^+,
\end{equation}
which completes the proof.
\end{proof}

\end{document}

%% file: weightsetup.tex
\begin{pspicture}[showgrid=false](0,0)(8,3)
\newgray{asv}{0.9}
\newgray{darker}{0.3}
\pspolygon[fillstyle=vlines,hatchcolor=asv,linestyle=none](0,0)(0,1.7)(4,1.7)(4,1.75)(8,1.75)(8,0)
\pspolygon[fillstyle=hlines,hatchcolor=asv,linestyle=none](0,3)(0,1.8)(4,1.8)(4,1.75)(8,1.75)(8,3)
\psline[linestyle=dashed](4,1.75)(8,1.75)
 \qline(0,0)(8,0)
 \qline(0,3)(8,3)
\psline(0,1.7)(4,1.7)(4,1.8)(0,1.8)
\psline[linecolor=darker]{<->}(6,1.75)(6,3)
\psline[linecolor=darker]{<->}(6,0)(6,1.75)
\rput[lm](6.1,0.875){$H_2$}
\rput[lm](6.1,2.375){$H_1$}
\psline[linecolor=gray]{->}(4,1.75)(4,2.75)
\psline[linecolor=gray]{->}(4,1.75)(5,1.75)
\rput[tm](4.9,1.6){\textcolor{gray}{$X$}}
\rput[rm](3.9,2.6){\textcolor{gray}{$Y$}}
\end{pspicture}

%% file: elementaryblock.tex
\psset{unit=0.8cm}
\begin{pspicture}[showgrid=false](-6,-3.3)(7,1.5)
\newgray{asv}{0.9}
\newgray{darker}{0.3}
\qline(-6,-2)(6,-2)
\qline(-6,1.5)(6,1.5)
\qline(-6,0)(-3,0)
\qline(3,0)(6,0)
\qline(6,1.5)(6,-2)
\psline{<->}(-6,-2)(-6,0)
\rput[rm](-6.1,-1){$\varepsilon H_2$}
\psline{<->}(-6,0)(-6,1.5)
\rput[rm](-6.1,0.75){$\varepsilon H_1$}
\pscurve(-3,0)(0,0.35)(3,0)
\pscurve(-3,0)(0,-0.35)(3,0)
\psline[linecolor=gray]{->}(0,0)(0,1)
\psline[linecolor=gray]{->}(0,0)(1,0)
\rput[tm](0.9,-0.1){\textcolor{gray}{$x$}}
\rput[rm](-0.1,0.9){\textcolor{gray}{$y$}}
\psline[linestyle=dotted,linecolor=darker](-3,-2.5)(-3,1.5)
\psline[linestyle=dotted,linecolor=darker](3,-2.5)(3,1.5)
\psline[linecolor=darker]{<->}(-3,-2.5)(3,-2.5)
\rput[tm](0,-2.6){$l$}
\psline[linestyle=dotted,linecolor=darker](-6,-3.1)(-6,-2)
\psline[linestyle=dotted,linecolor=darker](6,-3.1)(6,-2)
\psline[linecolor=darker]{<->}(-6,-3.1)(6,-3.1)
\rput[tm](0,-3.2){$a$}
\rput(-4.5,0.3){$\Omega^{(1)}_\varepsilon$}
\rput(1.5,0.8){$\Omega^{(2)}_\varepsilon$}
\rput(1.5,-1){$\Omega^{(3)}_\varepsilon$}
\rput(4.5,0.3){$\Omega^{(4)}_\varepsilon$}
\rput(-3.2,-0.3){$A$}
\rput(3.2,-0.3){$B$}
\psline[linecolor=gray](6.1,1.5)(6.3,1.5)(6.3,0.05)(6.1,0.05)
\psline[linecolor=gray](6.1,-0.05)(6.3,-0.05)(6.3,-2)(6.1,-2)
\rput[lm](6.4,0.75){\textcolor{gray}{$\Pi_\varepsilon^{(1)}$}}
\rput[lm](6.4,-1){\textcolor{gray}{$\Pi_\varepsilon^{(2)}$}}
\end{pspicture}

%% file: contourofintegration.tex
\begin{pspicture}[showgrid=false](-7,-4)(4,4)
\psset{arrowscale=2}
\psline{->}(3,0.2)(3,1.75)
\qline(3,1.6)(3,3)
\psline{->}(3,3)(-1.5,3)
\qline(-1.4,3)(-6,3)
\psline{->}(-6,3)(-6,1.75)
\qline(-6,1.9)(-6,0.2)
\psline{->}(-6,0.2)(-3,0.2)
\qline(-4,0.2)(-0.96,0.2)
\psarcn{->}(0,0){1}{168.463041}{90}
\psarcn(0,0){1}{96}{11.536959}
\psline{->}(0.96,0.2)(1.8,0.2)
\qline(1.7,0.2)(3,0.2)
\psline(3,-0.2)(3,-1.9)
\psline{->}(3,-3)(3,-1.75)
\psline(3,-3)(-1.7,-3)
\psline{->}(-6,-3)(-1.5,-3)
\psline(-6,-3)(-6,-1.6)
\psline{->}(-6,-0.2)(-6,-1.75)
\psline(-6,-0.2)(-3,-0.2)
\psline{->}(-0.96,-0.2)(-3.2,-0.2)
\psarcn(0,0){1}{-90}{-168.463041}
\psarcn{->}(0,0){1}{-11.536959}{-96}
\psline(0.96,-0.2)(2,-0.2)
\psline{->}(3,-0.2)(1.8,-0.2)
\pscircle*(0,0){0.05}
\rput[lm](0.3,0){$B$}
\rput(-1.5,3.4){$l_1^{(1)}$}
\rput(-1.5,-3.4){$l_1^{(2)}$}
\rput(-6.4,1.5){$l_2^{(1)}$}
\rput(-6.4,-1.5){$l_2^{(2)}$}
\rput(3.4,1.5){$l_5^{(1)}$}
\rput(3.4,-1.5){$l_5^{(2)}$}
\rput(-2.5,0.6){$l_3^{(1)}$}
\rput(-2.5,-0.6){$l_3^{(2)}$}
\rput(2,0.6){$l_4^{(1)}$}
\rput(2,-0.6){$l_4^{(2)}$}
\rput(0,1.4){$S_\delta^{(1)}$}
\rput(0,-1.4){$S_\delta^{(2)}$}
\rput(-4,1.5){$\Pi_B^{(1)}{(L)}$}
\rput(-4,-1.5){$\Pi_B^{(2)}{(L)}$}
\rput[lm](-6,-3.3){$X_B=-L$}
\rput[rm](3,-3.3){$X_B=L$}
\end{pspicture}